\documentclass[a4paper,11pt,british]{article}
\usepackage{babel}

\usepackage[left=1.5in,right=1.5in,top=1.5in,bottom=1.5in]{geometry}

\usepackage{graphicx}
\usepackage{amsmath}
\usepackage{array}
\usepackage{url}
\usepackage{amssymb}
\usepackage{amsfonts}
\usepackage{booktabs}
\usepackage{hyperref}
\usepackage{orcidlink}
\usepackage[caption=false]{subfig}

\usepackage[ruled,vlined]{algorithm2e}
\SetKw{KwPP}{Public Parameters:}
\usepackage{mathtools}
\DeclarePairedDelimiter\ceil{\lceil}{\rceil}
\DeclarePairedDelimiter\floor{\lfloor}{\rfloor}
\SetKwComment{Comment}{\scriptsize $\triangleright$\ }{}
\makeatletter\def\dftbig#1{{\hbox{$\left#1\vbox to10\p@{}\right.\n@space$}}}\makeatother
\renewcommand{\Pr}{\mathsf{Pr}}
\newcommand{\E}{\mathsf{E}}
\newcommand{\eat}[1]{}

\usepackage{amsthm}
\usepackage{thmtools}
\newtheorem{thm}{Theorem}[section]
\newtheorem{lem}[thm]{Lemma}
\newtheorem{cor}[thm]{Corollary}
\theoremstyle{definition}

\begin{document}

\title{Aggregation and Transformation of \\ Vector-Valued Messages in the Shuffle \\ Model of Differential Privacy}

\author{Mary Scott\textsuperscript{\orcidlink{0000-0003-0799-5840}},
Graham Cormode\textsuperscript{\orcidlink{0000-0002-0698-0922}}, and
Carsten Maple\textsuperscript{\orcidlink{0000-0002-4715-212X}}}
\date{}

\maketitle

\begin{abstract}
Advances in communications, storage and computational technology allow significant quantities of data to be collected and processed by distributed devices.
Combining the information from these endpoints can realize significant societal benefit but presents challenges in protecting the privacy of individuals, especially important in an increasingly regulated world.
Differential privacy (DP) is a technique that provides a rigorous and provable privacy guarantee for aggregation and release.
The Shuffle Model for DP has been introduced to overcome challenges regarding the accuracy of local-DP algorithms and the privacy risks of central-DP.
In this work we introduce a new protocol for vector aggregation in the context of the Shuffle Model.
The aim of this paper is twofold; first, we provide a single message protocol for the summation of real vectors in the Shuffle Model, using advanced composition results.
Secondly, we provide an improvement on the bound on the error achieved through using this protocol through the implementation of a Discrete Fourier Transform, thereby minimizing the initial error at the expense of the loss in accuracy through the transformation itself.
This work will further the exploration of more sophisticated structures such as matrices and higher-dimensional tensors in this context, both of which are reliant on the functionality of the vector case.
\end{abstract}

\section{Introduction} \label{sec:intro}
The benefit of processing data from distributed sources is being realized in a range of applications in areas including medical diagnoses and treatment, transportation, agile manufacturing, utilities management and entertainment services.
The rapid adoption of Internet of Things (IoT) systems that leverage recent advances in information collection, processing, communication and analysis, has played a significant role in realizing these benefits.
However, much of the information collected in IoT systems can, directly or indirectly, reveal personal information of the parties involved.
Such privacy concerns are gaining importance and concern in an increasingly regulated space.
Differential Privacy (DP)~\cite{dworkintro} has emerged as the leading candidate to provide privacy protection in the mining and release of private data.
DP provides a strong, mathematical definition of privacy that guarantees a measurable level of confidentiality for any data subject in the dataset to which it is applied.
In this way, useful collective information can be learned about a population, whilst simultaneously protecting the personal information of each data subject.

In particular, DP guarantees that the impact on any particular individual as a result of analysis on a dataset is the same, whether or not the individual is included in the dataset.
This guarantee is quantified by a parameter $\varepsilon$, which reflects strong privacy in cases where it is small.
However, finding an algorithm that achieves DP often requires a trade-off between privacy and accuracy: a smaller $\varepsilon$ sacrifices accuracy for better privacy, and vice versa.
DP enables data analyses such as the statistical analysis of the salaries of a population.
This allows useful collective information to be studied, so long as $\varepsilon$ is adjusted appropriately to satisfy the definition of DP.

In this work we focus on protocols in the \emph{Single-Message Shuffle Model}~\cite{balleprivacyblanket}, a one-time data collection model where each of $n$ users is permitted to submit a single message.
However, this constraint of applying to single messages restricts the applicability of the model.
We address this by applying the Single-Message Shuffle Model to the problem of \emph{vector aggregation}.
This is a valuable contribution since there are an increasing number of use cases, including Federated Learning, that utilize vector aggregation.

There are many practical applications of the Single-Message Shuffle Model in this federated setting, where multiple users collaboratively solve a Machine Learning problem, the results of which simultaneously improves the model for the next round~\cite{mcmahan}.
The updates generated by the users after each round are high-dimensional vectors, so this data type will prove useful in applications such as training a Deep Neural Network to predict the next word that a user types~\cite{neural}.
It is feasible to implement our shuffle-based protocol within the framework of Secure Aggregation, which would remove the requirement for an explicit entity to perform the shuffle~\cite{bonawitz}.

Our first contribution is a new protocol in the Single-Message Shuffle Model for the private summation of vector-valued messages, extending an existing result from Balle \emph{et al.}~\cite{balleprivacyblanket} by permitting the $n$ users to each submit a vector of real numbers rather than being restricted to submitting a scalar.
The resulting estimator is unbiased and has normalized mean squared error (MSE) $O_{\varepsilon, \delta} (d^{8/3} n^{-5/3})$, where $d$ is the dimension of each vector.
Our second contribution, which we call the Fourier Summation Algorithm (FSA), combines the private summation protocol with the Discrete Fourier Transform (DFT) from Rastogi and Nath in the centralized case~\cite{rastogi}, to improve the accuracy of the tight bound to $O_{\varepsilon, \delta} (m^{8/3} n^{-5/3})$, where $m$ represents the number of Fourier coefficients retained.
Since $m \ll d$, this is a considerable improvement on the previous estimator, though some accuracy is lost through the transformation of the messages between the original and Fourier domains.

Compared to prior work on scalar aggregation (sum), our work requires several innovations.
It begins with the same generalization of randomized response to encode each real input value into a discrete histogram as has been used in several prior works.
However, we then have to argue how to combine the results from multiple vector coordinates to rebuild a representation of the aggregate input.
Naively, it might seem that we have to divide the `privacy budget' ($\epsilon$) into $d$ pieces to process a $d$-dimensional histogram.
However, our analysis shows that this can be improved so that we sample $t$ out of $d$ locations in the vector, where the privacy cost only scales proportional to $\sqrt{t}$; moreover, we show that analytically and empirically it is best to set $t$ as small as possible, i.e., to sample $t=1$ coordinates from each client.
We introduce the idea of combining the Fourier transformation with privacy in the shuffle model, and demonstrate that it is possible to improve the accuracy/communication trade-off, by sampling from a reduced selection of Fourier coefficients.
It is not meaningful to apply such a transformation in the scalar case, and so the approach is new to the vector setting.

It is possible for these vector summation protocols to be extended to produce a similar protocol for the linearization of matrices.
To do this, it must be recognized that matrix decomposition or reduction is required to ensure that the constituent vectors are linearly independent.
Given that we fix the dimension of each matrix, it is important to ensure that all constituent vectors are linearly independent, as this guarantees a unique solution for each matrix.
Our method for matrices can be further extended to higher-dimensional tensors, which are useful for the representation of multi-dimensional data in Neural Networks.

\section{Related Work} \label{sec:litreview}

The earliest attempts at protecting the privacy of users in a dataset focused on simple ways of suppressing or generalizing the data.
Examples include $k$-anonymity~\cite{kanonymity}, $l$-diversity~\cite{ldiversity} and $t$-closeness~\cite{tcloseness}.
However, such attempts have been shown to be insufficient, as proved by numerous examples~\cite{dwork}.

This harmful leakage of sensitive information can be prevented through the application of DP, since the method mathematically guarantees that the chance of a \emph{linkage attack} on an individual in the dataset is almost identical to that on an individual not in the dataset.

Since DP was first conceptualized in 2006 by Dwork \emph{et al.}~\cite{dworkintro}, the majority of research in the field has focused on two contrasting models.
In the Centralized Model, users submit their sensitive personal information directly to a \emph{trusted} central data collector, who adds \emph{random noise} to the raw data to provide DP, before assembling and analyzing the aggregated results.

In the Local Model, DP is guaranteed when each user applies a \emph{local randomizer} to add random noise to their data before it is submitted.
The Local Model differs from the Centralized Model in that the central entity does not see the users' raw data at any point, and therefore does not have to be trusted.
However, the level of noise required per user for the same privacy guarantee is much higher, limiting the efficacy of Local Differential Privacy (LDP) unless used in very large populations.
For this reason, the application of LDP is largely the domain of major companies such as Google~\cite{google}, Apple~\cite{apple} and Microsoft~\cite{microsoft}.

Neither of the two models can provide a good balance between the trust of the central entity and the level of noise required to guarantee DP.
Hence, in recent years researchers have tried to create intermediate models that reap the benefits of both.

In 2017, Bittau \emph{et al.}~\cite{bittau} introduced the Encode, Shuffle, Analyze (ESA) model, which provides a general framework for the addition of a \emph{shuffling step} in a private protocol.
After the data from each user is encoded, it is randomly permuted to unbind each user from their data before analysis takes place.
In 2019, Cheu \emph{et al.}~\cite{cheu} formalized the Shuffle Model as a special case of the ESA model; their model connects this additional shuffling step to the Local Model.
In the Shuffle Model, the local randomizer applies a randomized mechanism on a per-element basis, potentially replacing a truthful value with another randomly selected domain element.
The role of these independent reports is to create what is known as a \emph{privacy blanket}, which masks the outputs which are reported truthfully.

As well as the result on the private summation of scalar-valued messages in the Single-Message Shuffle Model that we will be using~\cite{balleprivacyblanket}, Balle \emph{et al.} have published two more recent works that solve related problems.
The first paper~\cite{balleimproved} improved the distributed $n$-party summation protocol from Ishai \emph{et al.}~\cite{ishai} in the context of the Single-Message Shuffle Model to require $O (1 + \pi/\log n)$ scalar-valued messages, instead of a logarithmic dependency of $O (\log n + \pi)$, to achieve statistical security $2^{-\pi}$.
The second paper~\cite{ballemulti} introduced two new protocols for the private summation of scalar-valued messages in the Multi-Message Shuffle Model, an extension of the Single-Message Shuffle Model that permits each of the $n$ users to submit more than one message, using several independent shufflers to securely compute the sum.
In this work, Balle \emph{et al.} contributed a recursive construction based on the protocol described in~\cite{balleprivacyblanket}, as well as an alternative mechanism which implements a discretized distributed noise addition technique using the result from Ishai \emph{et al.}~\cite{ishai}.

A relevant concurrent work to our first contribution is the work of Girgis \emph{et al.}~\cite{girgis}, which uses the Single-Message Shuffle Model directly in the Federated Learning framework.
This contrasts with the link of our contribution to Federated Learning as a use case of vector aggregation.
A recent paper by Feldman \emph{et al.}~\cite{clones} extends the `amplification by shuffling' problem: the remaining result of Balle \emph{et al.}~\cite{balleprivacyblanket} that is outside the scope of our work.

Also relevant to our research is the work of Ghazi \emph{et al.}~\cite{ghazipower}, which explored the related problems of private frequency estimation and selection in a similar context, drawing comparisons between the errors achieved in the Single-Message Shuffle Model and the Multi-Message Shuffle Model.
A similar team of authors produced a follow-up paper~\cite{ghazidistributed} describing a more efficient protocol for private summation in the Single-Message Shuffle Model, using the `invisibility cloak' technique to facilitate the addition of zero-sum noise without coordination between the users.
The most recent work of Ghazi \emph{et al.}~\cite{ghazimissing} relaxes the single-message requirement of their previous protocols to improve the accuracy of private summation in the Shuffle Model to be close to that of the Centralized Model.

Several related works have provided inspiration for our design employing the Discrete Fourier Transform (DFT) for private summation in the Single-Message Shuffle Model.
Rastogi and Nath~\cite{rastogi} introduced the idea of using a Fourier transform in the central privacy model in order to reduce the aggregate amount of privacy noise added; here, our contribution is to show a corresponding result in the shuffle model.
Selesnick \emph{et al.}~\cite{selesnick} describe numerous symmetric extensions to the DFT, each of which guaranteed a real-valued output for a real-valued input.
This proved useful for our protocol, since the representation of their data as a vector in a high-dimensional space is closely related to the representation of our data as vector-valued messages.
Finally, Cormode \emph{et al.}~\cite{cormode} explored the application of the DFT over the Boolean hypercube, also known as the Hadamard Transform, in the Local Model.
Their algorithms provide a useful link between the theory of the DFT and its application to a closely related model of DP, as well as illustrating the benefits of such a transform on the resulting dependencies.

\section{Preliminaries} \label{sec:prelim}

We consider randomized mechanisms~\cite{dwork} $\mathcal{M}$, $\mathcal{R}$ under domains $\mathbb{X}$, $\mathbb{Y}$, and apply them to input datasets $\vec{D}, \vec{D}'$ to generate (vector-valued) messages $\vec{x}_{i}, \vec{x}_{i}'$.
We write $[k] = \{ 1, \dots, k \}$ and $\mathbb{N}$ for the set of natural numbers.

\subsection{Models of Differential Privacy} \label{sec:dp}

The essence of Differential Privacy (DP) is the requirement that the contribution $\vec{x}_{i}$, of a user $i$, to a dataset $\vec{D} = (\vec{x}_{1}, \dots, \vec{x}_{n})$ does not have a significant impact on the outcome of the mechanism applied to that dataset.

Let us consider the \emph{centralized} model of DP, in which random noise is only introduced after the users' inputs are gathered by a (trusted) aggregator.
Consider further a dataset $\vec{D}'$ that differs from $\vec{D}$ only in the contribution of a single user, denoted $\vec{D} \simeq \vec{D}'$.
Given $\varepsilon \geq 0$ and $\delta \in (0,1)$, we define  a randomized mechanism $\mathcal{M}: \mathbb{X}^{n} \rightarrow \mathbb{Y}$ to be $(\varepsilon, \delta)$-differentially private if $\forall \vec{D} \simeq \vec{D}', \forall E \subseteq \mathbb{Y}$:
\[
\Pr[\mathcal{M}(\vec{D}) \in E] \leq e^{\varepsilon} \cdot \Pr[\mathcal{M}(\vec{D}') \in E] + \delta~\cite{dwork}.
\]

In this definition, we assume that the trusted aggregator obtains raw data from all users and introduces necessary mechanisms to provide privacy.

In the \emph{local} model of DP, each user $i$ independently uses randomness on their input $\vec{x}_{i} \in \mathbb{X}$ by using a \emph{local randomizer} $\mathcal{R}: \mathbb{X} \rightarrow \mathbb{Y}$ to obtain a perturbed result $\mathcal{R}(\vec{x}_{i})$.
We say that the local randomizer is $(\varepsilon, \delta)$-differentially private if $\forall \vec{D}, \vec{D}', \forall E \subseteq \mathbb{Y}$:
\begin{align*}
\Pr[\mathcal{R}(\vec{x}_{i}) \in E] \leq e^{\varepsilon} \cdot \Pr[\mathcal{R}(\vec{x}_{i}') \in E] + \delta~\cite{balleprivacyblanket}, 
\end{align*}

\noindent where $\vec{x}_{i}' \in \mathbb{X}$ is some other valid input vector that $i$ could hold.
The Local Model guarantees that any observer will not have access to the raw data from any of the users.
That is, it removes the requirement for trust in the aggregator.
The consequence of this removal of trust is that a higher level of noise per user must be tolerated to achieve the same privacy guarantee.

\subsection{Single-Message Shuffle Model} \label{sec:smsm}

The Single-Message Shuffle Model can be considered to sit in between the Centralized and Local Models of DP~\cite{balleprivacyblanket}.
Let a protocol $\mathcal{P}$ in the Single-Message Shuffle Model be of the form $\mathcal{P} = (\mathcal{R}, \mathcal{A})$, where $\mathcal{R}: \mathbb{X} \rightarrow \mathbb{Y}$ is the \emph{local randomizer}, and $\mathcal{A}: \mathbb{Y}^{n} \rightarrow \mathbb{Z}$ is the \emph{analyzer} of $\mathcal{P}$.

Overall, $\mathcal{P}$ implements a mechanism $\mathcal{P}: \mathbb{X}^{n} \rightarrow \mathbb{Z}$ as follows.
Each user $i$ independently applies the local randomizer to their message $\vec{x}_{i}$ to obtain a message $\vec{y}_{i} = \mathcal{R}(\vec{x}_{i})$.
Subsequently, the messages $(\vec{y}_{1}, \dots, \vec{y}_{n})$ are randomly permuted by a trusted \emph{shuffler} $\mathcal{S}: \mathbb{Y}^{n} \rightarrow \mathbb{Y}^{n}$.
The random permutation $\mathcal{S}(\vec{y}_{1}, \dots, \vec{y}_{n})$ is submitted to an untrusted data collector, who applies the analyzer $\mathcal{A}$ to obtain an output for the mechanism.
In summary, the output of $\mathcal{P}(\vec{x}_{1}, \dots, \vec{x}_{n})$ is given by:
\[
\mathcal{A} \circ \mathcal{S} \circ \mathcal{R}^{n}(\vec{x}) = \mathcal{A}(\mathcal{S}(\mathcal{R}(\vec{x}_{1}), \dots, \mathcal{R}(\vec{x}_{n}))).
\]

Note that the data collector observing the shuffled messages $\mathcal{S}(\vec{y}_{1}, \dots, \vec{y}_{n})$ obtains no information about which user generated each of the messages.
Therefore, the privacy of $\mathcal{P}$ relies on the indistinguishability between the shuffles $\mathcal{S} \circ \mathcal{R}^{n}(\vec{D})$ and $\mathcal{S} \circ \mathcal{R}^{n}(\vec{D}')$ for datasets $\vec{D} \simeq \vec{D}'$.
The analyzer can represent the shuffled messages as a \emph{histogram}, which counts the number of occurrences of the possible outputs of $\mathbb{Y}$.

\subsection{Measuring Accuracy} \label{sec:mse}

In Sections~\ref{sec:vectorsum} and \ref{sec:transform} we use the \emph{mean squared error} to compare the overall output of our new private summation protocol in the Single-Message Shuffle Model with the original dataset.
The MSE is used to measure the average squared difference in the comparison between a fixed input $f(\vec{D})$ to the randomized protocol $\mathcal{P}$, and its output $\mathcal{P}(\vec{D})$.
In this context,
\[
\text{MSE}(\mathcal{P}, \vec{D}) = \E \dftbig[ (\mathcal{P}(\vec{D}) - f(\vec{D}))^{2} \dftbig],
\]
where the expectation is taken over the randomness of $\mathcal{P}$.
Note when $\E[\mathcal{P}(\vec{D})] =  f(\vec{D})$, MSE is equivalent to variance, i.e.:
\[
\text{MSE}(\mathcal{P}, \vec{D}) = \E \dftbig[ (\mathcal{P}(\vec{D}) - \E[\mathcal{P}(\vec{D})])^{2} \dftbig] = \text{Var}[\mathcal{P}(\vec{D})].
\]

\section{Vector Sum in the Shuffle Model} \label{sec:vectorsum}

In this section we introduce our new protocol for vector summation in the Shuffle Model and tune its parameters to optimize accuracy.

\subsection{Basic Randomizer} \label{sec:basic}

First, we describe a basic local randomizer applied by each user $i$ to an input $x_i \in [k]$, a fundamental technique in privacy.
The output of this protocol is a (private) histogram of shuffled messages over the domain $[k]$.

The Local Randomizer $\mathcal{R}_{\gamma, k, n}^{PH}$, shown in Algorithm~\ref{alg:localrand}, applies a generalized \emph{randomized response} mechanism that returns the true message $x_i$ with probability $1 - \gamma$ and a uniformly random message with probability $\gamma$.
Such a basic randomizer is used by Balle \emph{et al.}~\cite{balleprivacyblanket} in the Single-Message Shuffle Model for scalar-valued messages, as well as in several other previous works in the Local Model~\cite{kairouzextremal, kairouzdiscrete, bhowmick}.
In Section~\ref{sec:privanalysis}, we find an appropriate $\gamma$ to optimize the proportion of random messages that are submitted, and therefore guarantee DP.

\begin{algorithm}[t]
\DontPrintSemicolon
\SetArgSty{textnormal}
\KwPP{\parbox[t]{2.1in}{\raggedright $\gamma \in [0,1]$, domain size $k$, and number of parties $n$}}\\
\KwIn{$x_i \in [k]$}
\KwOut{$y_i \in [k]$}
Sample $b \leftarrow$ {\tt Ber}$(\gamma)$\\
\lIf{$b=0$}{let $y_i \leftarrow x_i$}
\lElse{sample $y_i \leftarrow$ {\tt Unif}$([k])$}
\KwRet{$y_i$}
\caption{Local Randomizer $\mathcal{R}_{\gamma, k, n}^{PH}$}
\label{alg:localrand}
\end{algorithm}

We now describe how the presence of these random messages can form a `privacy blanket' to protect against a \emph{difference attack} on a particular user.
Suppose we apply Algorithm~\ref{alg:localrand} to the messages from all $n$ users.
Note that a subset $B$ of approximately $\gamma n$ of these users returned a uniformly random message, while the remaining users returned their true message.
Following Balle \emph{et al.}~\cite{balleprivacyblanket}, the analyzer can represent the messages sent by users in $B$ by a histogram $Y_1$ of uniformly random messages, and can form a histogram $Y_2$ of truthful messages from users not in $B$.
As these subsets are mutually exclusive and collectively exhaustive, the information represented by the analyzer is equivalent to the histogram $Y = Y_1 \cup Y_2$.

Consider two neighbouring datasets, each consisting of $n$ messages from $n$ users, that differ only on the input from the $n^{\text{th}}$ user.
To simplify the discussion and subsequent proof, we temporarily omit the action of the shuffler.
By the post-processing property of DP, this can be reintroduced later on without adversely affecting the privacy guarantees.
To achieve DP we need to find an appropriate $\gamma$ such that when Algorithm~\ref{alg:localrand} is applied, the change in $Y$ is appropriately bounded.
As the knowledge of either the set $B$ or the messages from the first $n - 1$ users does not affect DP, we can assume that the analyzer knows both of these details.
This lets the analyzer remove all of the truthful messages associated with the first $n - 1$ users from $Y$.

If the $n^{\text{th}}$ user is in $B$, this means their submission is independent of their input, so we trivially satisfy DP.
Otherwise, the (curious) analyzer knows that the $n^{\text{th}}$ user has submitted their true message $x_n$.
The analyzer can remove all of the truthful messages associated with the first $n-1$ users from $Y$, and obtain $Y_1 \cup \{ x_n \}$.
The subsequent privacy analysis will argue that this does not reveal $x_n$ if $\gamma$ is set so that $Y_1$, the histogram of random messages, appropriately `hides' $x_n$.

\subsection{Private Summation of Vector-Valued Messages} \label{sec:protocol}

Here, we extend the protocol from Section~\ref{sec:basic} to address the problem of computing the sum of $n$ real vectors, each of the form $\vec{x}_{i} = (x_{i}^{(1)}, \dots, x_{i}^{(d)}) \in [0,1]^{d}$, in the Single-Message Shuffle Model.
Specifically, we analyze the utility of a protocol $\mathcal{P}_{d, k, n, t} = (\mathcal{R}_{d, k, n, t}, \mathcal{A}_{d, k, t})$ for this purpose, by using the MSE from Section~\ref{sec:mse} as the accuracy measure.
In the scalar case, each user applies the protocol to their entire input~\cite{balleprivacyblanket}.
Moving to the vector case, we allow each user to independently sample a set of $1 \le t \le d$ coordinates from their vector to report.
Our analysis allows us to optimize the parameter $t$.

Hence, the first step of the Local Randomizer $\mathcal{R}_{d, k, n, t}$, presented in Algorithm~\ref{alg:fixedpoint}, is to uniformly sample $t$ coordinates $(\alpha_{i1}, \dots, \alpha_{it}) \in [d]$ (without replacement) from each vector $\vec{x}_{i}$.
To compute a differentially private approximation of $\sum_{i} \vec{x}_{i}$, we fix a quantization level $k$.
Then we randomly round each $x_{i}^{(\alpha_{ij})}$ to obtain $\bar{x}_{i}^{(\alpha_{ij})}$ as either $\floor{{x}_{i}^{(\alpha_{ij})} k}$ or $\ceil{{x}_{i}^{(\alpha_{ij})} k}$.
Next, we apply the randomized response mechanism from Algorithm~\ref{alg:localrand} to each $\bar{x}_{i}^{(\alpha_{ij})}$, which sets each output $y_{i}^{(\alpha_{ij})}$ independently to be equal to $\bar{x}_{i}^{(\alpha_{ij})}$ with probability $1 - \gamma$, or a random value in $\{ 0, 1, \dots, k \}$ with probability $\gamma$.
Each $y_{i}^{(\alpha_{ij})}$ will contribute to a histogram of the form $(y_{1}^{(\alpha_{ij})}, \dots, y_{n}^{(\alpha_{ij})})$ as in Section~\ref{sec:basic}.

The Analyzer $\mathcal{A}_{d, k, t}$, shown in Algorithm~\ref{alg:analyzer}, aggregates the histograms to approximate $\sum_{i} \vec{x}_{i}$ by post-processing the vectors coordinate-wise.
More precisely, the analyzer sets each output $y_{i}^{(\alpha_{ij})}$ to $y_{i}^{(l)}$, where the new label $l$ is from its corresponding input $x_{i}^{(l)}$ of the original $d$-dimensional vector $\vec{x}_{i}$.
For all inputs $x_{i}^{(l)}$ that were not sampled, we set $y_{i}^{(l)} = 0$.
Subsequently, the analyzer aggregates the sets of outputs from all users corresponding to each of those $l$ coordinates in turn, so that a $d$-dimensional vector is formed.
Finally, a standard debiasing step is applied to this vector to remove the scaling and rounding applied to each submission.
\texttt{DeBias} returns an unbiased estimator, $\vec{z}$, which calculates an estimate of the true sum of the vectors by subtracting the expected uniform noise from the randomized sum of the vectors.

Note that Algorithms~\ref{alg:fixedpoint} and \ref{alg:analyzer} are both required to generalize the scalar approach from Balle \emph{et al.}~\cite{balleprivacyblanket} to vectors.
In Section~\ref{sec:privanalysis}, we carefully prove that we can combine Algorithms~\ref{alg:fixedpoint} and \ref{alg:analyzer} to privately compute the sum of vector-valued messages in the Shuffle Model, thus resulting in our first contribution.

\begin{algorithm}[t]
\DontPrintSemicolon
\SetArgSty{textnormal}
\KwPP{$k$, $t$, dimension $d$, and number of parties $n$}\\
\KwIn{$\vec{x}_{i} = (x_{i}^{(1)}, \dots, x_{i}^{(d)}) \in [0,1]^{d}$}
\KwOut{$\vec{y}_{i} = (y_{i}^{(\alpha_{i1})}, \dots, y_{i}^{(\alpha_{it})}) \in \{0,1, \dots, k\}^{t}$}
\medskip
Sample $(\alpha_{i1}, \dots, \alpha_{it}) \leftarrow$ {\tt Unif}$([d])$\\
Let $\bar{x}_{i}^{(\alpha_{ij})} \leftarrow \floor[\dftbig]{x_{i}^{(\alpha_{ij})} k} \ +$ {\tt Ber}$ (x_{i}^{(\alpha_{ij})} k - \floor[\dftbig]{x_{i}^{(\alpha_{ij})} k})$
\Comment*[r]{\parbox[t]{4in}{\raggedright $\bar{x}_{i}^{(\alpha_{ij})}$: encoding of $x_{i}^{(\alpha_{ij})}$ with precision $k$}}
\vspace*{-.3cm}
\Comment*[r]{\parbox[t]{4in}{\raggedright $y_{i}^{(\alpha_{ij})}$: apply \textbf{Algorithm 1} to each $\bar{x}_{i}^{(\alpha_{ij})}$}}
\KwRet{$\vec{y}_{i} = (y_{i}^{(\alpha_{i1})}, \dots, y_{i}^{(\alpha_{it})})$}
\caption{Local Randomizer $\mathcal{R}_{d, k, n, t}$}
\label{alg:fixedpoint}
\end{algorithm}

\begin{algorithm}[t]
\DontPrintSemicolon
\SetArgSty{textnormal}
\KwPP{$k$, $t$, and dimension $d$}\\
\KwIn{Multiset $\bigl\{ \vec{y}_{i} \bigr\}_{i \in [n]}$, with $(y_{i}^{(\alpha_{i1})}, \dots, y_{i}^{(\alpha_{it})}) \in \{0,1, \dots, k\}^{t}$}
\KwOut{$\vec{z} = (z^{(1)}, \dots, z^{(d)}) \in [0,1]^{d}$}
\medskip
Let $y_{i}^{(l)} \leftarrow y_{i}^{(\alpha_{ij})}$
\medskip
\Comment*[r]{\parbox[t]{4.3in}{\raggedright $y_{i}^{(\alpha_{ij})}$: submission corresponding to $x_{i}^{(l)}$}}
Let $(\hat{z}^{(1)}, \dots, \hat{z}^{(d)}) \leftarrow (\frac{1}{k} \sum_{i} y_{i}^{(1)}, \dots, \frac{1}{k} \sum_{i} y_{i}^{(d)})$\\
Let $(z^{(1)}, \dots, z^{(d)}) \leftarrow (${\tt DeBias}$ (\hat{z}^{(1)}), \dots, ${\tt DeBias}$ (\hat{z}^{(d)}))$
\Comment*[r]{\parbox[t]{4in}{\raggedright {\tt DeBias}$(\hat{z}^{(l)}) = ( \hat{z}^{(l)} - \frac{\gamma}{2} \cdot | y_{i}^{(l)}|) / (1 - \gamma)$}}
\KwRet{$\vec{z} = (z^{(1)}, \dots, z^{(d)})$}
\caption{Analyzer $\mathcal{A}_{d, k, t}$}
\label{alg:analyzer}
\end{algorithm}

\subsection{Privacy Analysis of Algorithms~\ref{alg:fixedpoint} and \ref{alg:analyzer}} \label{sec:privanalysis}

In this section, we will find an appropriate $\gamma$ that ensures that the mechanism described in Algorithms~\ref{alg:fixedpoint} and \ref{alg:analyzer} satisfies $(\varepsilon, \delta)$-DP for vector-valued messages in the Single-Message Shuffle Model.
To achieve this, we prove the following theorem, where we initially assume $\varepsilon < 1$ to simplify our computations.

At the end of this section, we discuss how to cover the additional case $1 \leq \varepsilon < 6$ to suit our experimental study.
This moderate range of $\varepsilon$ is justified by the fact that privacy is weak for $\varepsilon \geq 6$.
The upper limit of $\epsilon$ is arbitrary: it can be set to any positive integer, with an almost identical proof in each case.
Therefore, we have chosen $6$ as the limit due to practical usage, as echoed by the literature~\cite{dwork,apple,microsoft}.

\begin{thm} \label{thm:main}
The shuffled mechanism $\mathcal{M} = \mathcal{S} \circ \mathcal{R}_{d, k, n, t}$ is $(\varepsilon, \delta)$-DP for any $d, k, n \in \mathbb{N}$, $\{ t \in \mathbb{N} \ | \ t \in [d] \}$, $\varepsilon < 6$ and $\delta \in \left( 0,1 \right]$ such that:
\[
\gamma = 
\begin{cases}
\frac{56dk \log(1/\delta) \log(2t/\delta)}{(n - 1) \varepsilon^{2}}, & \text{when}\ \varepsilon < 1 \\
\frac{2016dk \log(1/\delta) \log(2t/\delta)}{(n - 1) \varepsilon^{2}}, & \text{when}\ 1 \leq \varepsilon < 6. \\
\end{cases}
\]
\end{thm}

\begin{proof}
Let $\vec{D} = (\vec{x}_{1}, \dots, \vec{x}_{n})$ and $\vec{D}' = (\vec{x}_{1}, \dots, \vec{x}'_{n})$ be the two neighbouring datasets differing only in the input of the $n$\textsuperscript{th} user, as used in Section~\ref{sec:basic}.
Here each vector-valued message $\vec{x}_{i}$ is of the form $(x_{i}^{(1)}, \dots, x_{i}^{(d)})$.
Recall from Section~\ref{sec:basic} that we assume that the analyzer can see the users in $B$ (i.e., the subset of users that returned a uniformly random message), as well as the inputs from the first $n - 1$ users.

We now introduce the \emph{vector view} $\text{VView}_{\mathcal{M}}(\vec{D})$ as the collection of information that the analyzer is able to see after the mechanism $\mathcal{M}$ is applied to all vector-valued messages in the dataset $\vec{D}$.
$\text{VView}_{\mathcal{M}}(\vec{D})$ is defined as the tuple 
$(\vec{Y}, \vec{D}_{\cap}, \vec{b})$, where $\vec{Y}$ is the multiset containing the outputs $\{ \vec{y}_{1}, \dots, \vec{y}_{n} \}$ of the mechanism $\mathcal{M}(\vec{D})$, $\vec{D}_{\cap}$ is the vector containing the inputs $(\vec{x}_{1}, \dots, \vec{x}_{n - 1})$ from the first $n - 1$ users, and $\vec{b}$ contains binary vectors $(\vec{b}_{1}, \dots, \vec{b}_{n})$ which indicate for which coordinates each user reports truthful information.
This vector view can be projected to $t$ overlapping \emph{scalar views} by applying Algorithm~\ref{alg:fixedpoint} only to the $j^{\text{th}}$ uniformly sampled coordinate $\alpha_{ij} \in [d]$ from each user, where $j \in [t]$.
The $j^{\text{th}}$ scalar view $\text{View}_{\mathcal{M}}^{(\alpha_{ij})}(\vec{D})$ of $\text{VView}_{\mathcal{M}}(\vec{D})$ is defined as the tuple $(\vec{Y}^{(\alpha_{ij})}, \vec{D}_{\cap}^{(\alpha_{ij})},  \vec{b}^{(\alpha_{ij})})$, where:
\begin{align*}
\vec{Y}^{(\alpha_{ij})} & = \mathcal{M} (\vec{D}^{(\alpha_{ij})}) = \{ y_{1}^{(\alpha_{ij})}, \dots, y_{n}^{(\alpha_{ij})} \}, \\ \vec{D}_{\cap}^{(\alpha_{ij})} & = (x_{1}^{(\alpha_{ij})}, \dots, x_{n - 1}^{(\alpha_{ij})}) \\
\text{and} \quad \vec{b}^{(\alpha_{ij})} & = (b_{1}^{(\alpha_{ij})}, \dots, b_{n}^{(\alpha_{ij})})
\end{align*}
are the analogous definitions of $\vec{Y}$, $\vec{D}_{\cap}$ and $\vec{b}$, but containing only the information referring to the $j^{\text{th}}$ uniformly sampled coordinate of each vector-valued message.

The following \emph{advanced composition} results will be used in our setting to get a tight upper bound:

\begin{thm}[Dwork \emph{et al.}~\cite{dwork}] \label{thm:advcomp}
For all $\varepsilon', \delta', \delta \geq 0$, the class of $(\varepsilon', \delta')$-differentially private mechanisms satisfies $(\varepsilon, r \delta' + \delta)$-differential privacy under $r$-fold adaptive composition for:
\[
\varepsilon = \sqrt{2r \log(1/\delta)} \varepsilon' + r \varepsilon' \dftbig( e^{\varepsilon'} - 1 \dftbig).
\]
\end{thm}

\begin{cor} \label{cor:advcomp}
Given target privacy parameters $0 < \varepsilon < 1$ and $\delta > 0$, to ensure $(\varepsilon, r \delta' + \delta)$ cumulative privacy loss over $r$ mechanisms, it suffices that each mechanism is $(\varepsilon', \delta')$-DP, where:
\[
\varepsilon' = \frac{\varepsilon}{2 \sqrt{2r \log(1/\delta)}}.
\]
\end{cor}

\noindent To show that $\text{VView}_{\mathcal{M}}(\vec{D})$ satisfies $(\varepsilon, \delta)$-DP it suffices to prove that:
\[
\Pr_{\widetilde{\mathsf{V}} \sim \text{VView}_{\mathcal{M}}(\vec{D})} \!\left[ \frac{ \Pr [ \text{VView}_{\mathcal{M}}(\vec{D}) = \widetilde{\mathsf{V}} ] }
{ \Pr [ \text{VView}_{\mathcal{M}}(\vec{D}') = \widetilde{\mathsf{V}} ] } \geq e^{\varepsilon} \right] \leq \delta. \label{eq:vector} \tag{$1$}
\]

\noindent By considering this vector view as a union of overlapping scalar views, and letting $r = t$ in Corollary~\ref{cor:advcomp}, it is sufficient to derive \eqref{eq:vector} from:
\[
\Pr_{\mathsf{V}_{\alpha_{ij}} \sim \text{View}_{\mathcal{M}}^{(\alpha_{ij})}(\vec{D})} \!\left[ \frac{ \Pr [ \text{View}_{\mathcal{M}}^{(\alpha_{ij})}(\vec{D}) = \mathsf{V}_{\alpha_{ij}} ] }
{ \Pr [ \text{View}_{\mathcal{M}}^{(\alpha_{ij})}(\vec{D}') = \mathsf{V}_{\alpha_{ij}} ] } \geq e^{\varepsilon'} \right] \leq \delta', \label{eq:scalar} \tag{$2$}
\]

\noindent where $\widetilde{\mathsf{V}} = \bigcup_{\alpha_{ij}} \mathsf{V}_{\alpha_{ij}}$, $\varepsilon' = \frac{\varepsilon}{2 \sqrt{2t \log(1/\delta)}}$ and $\delta' = \frac{\delta}{t}$.

\begin{lem} \label{lem:implies}
Condition \eqref{eq:scalar} implies condition \eqref{eq:vector}.
\end{lem}

\begin{proof}
We can express $\text{VView}_{\mathcal{M}}(\vec{D})$ as the {composition} of the $t$ scalar views $\text{View}_{\mathcal{M}}^{(\alpha_{i1})}, \dots, \text{View}_{\mathcal{M}}^{(\alpha_{it})}$, as:
\begin{align*}
\Pr&[ \text{VView}_{\mathcal{M}}(\vec{D}) = \widetilde{\mathsf{V}} ] \\
&= \Pr [ \text{View}_{\mathcal{M}}^{(\alpha_{i1})}(\vec{D}) = \mathsf{V}_{\alpha_{i1}} \wedge \cdots \wedge \text{View}_{\mathcal{M}}^{(\alpha_{it})}(\vec{D}) = \mathsf{V}_{\alpha_{it}} ] \\
&= \Pr [ \text{View}_{\mathcal{M}}^{(\alpha_{i1})}(\vec{D}) = \mathsf{V}_{\alpha_{i1}} ] \boldsymbol{\cdot} \cdots \boldsymbol{\cdot} \Pr [ \text{View}_{\mathcal{M}}^{(\alpha_{it})}(\vec{D}) = \mathsf{V}_{\alpha_{it}}].
\end{align*}

Our desired result is immediate by applying Corollary~\ref{cor:advcomp}, which states that the use of $t$ overlapping $(\varepsilon', \delta')$-DP mechanisms, when taken together, is $(\varepsilon, \delta)$-DP.
This applies in our setting, since we have assumed that $\text{VView}_{\mathcal{M}}(\vec{D})$ satisfies the requirements of $(\varepsilon, \delta)$-DP, and that each of the $t$ overlapping scalar views is formed identically but for a different uniformly sampled coordinate of the vector-valued messages.
\end{proof}

To complete the proof of Theorem~\ref{thm:main} for $\varepsilon < 1$, it remains to show that for a uniformly sampled coordinate $\alpha_{ij} \in [d]$, $\text{View}_{\mathcal{M}}^{(\alpha_{ij})}(\vec{D})$ satisfies $(\varepsilon', \delta')$-DP.

\begin{restatable}{lem}{primelemma} \label{lem:scalar}
Condition \eqref{eq:scalar} holds.
\end{restatable}

\begin{proof}
See Appendix.
\end{proof}

We now show that the above proof can be adjusted to cover the additional case $1 \leq \varepsilon < 6$.
This will be sufficient to complete the proof of our main Theorem~\ref{thm:main}.

First, we scale the setting of $\varepsilon'$ by a multiple of $6$ in Corollary~\ref{cor:advcomp} so that the advanced composition property holds for all $1 \leq \varepsilon < 6$.
We now insert $\varepsilon' = \frac{\varepsilon}{12 \sqrt{2r \log(1/\delta)}}$ into the proof of Theorem~\ref{thm:main}, resulting in a change of constant from $56$ to $2016$.
\end{proof}

\subsection{Accuracy Bounds for Shuffled Vector Sum} \label{sec:tightupperbound}

We now formulate an upper bound for the MSE of our protocol, and then identify the value(s) of $t$ that minimize this upper bound.

First, note that encoding the coordinate $x_{i}^{(\alpha_{ij})}$ as $\bar{x}_{i}^{(\alpha_{ij})} = \floor[\dftbig]{x_{i}^{(\alpha_{ij})} k} \ +$ {\tt Ber}$ ( x_{i}^{(\alpha_{ij})} k - \floor[\dftbig]{x_{i}^{(\alpha_{ij})} k})$ in Algorithm~\ref{alg:fixedpoint} ensures that $\mathbb{E}[\bar{x}_{i}^{(\alpha_{ij})}/k] = \mathbb{E} [x_{i}^{(\alpha_{ij})}]$.
This means that our protocol is unbiased.
For any unbiased random variable $X$ with $a < X < b$ then $\text{Var}[X] \leq (b-a)^2/4$, and so the MSE per coordinate due to the fixed-point approximation of the true vector in $\mathcal{R}_{d, k, n, t}$  is at most $\frac{1}{4k^{2}}$.
Meanwhile, the MSE when $\mathcal{R}_{d, k, n, t}$ submits a random vector is at most $\frac{1}{2}$ per coordinate.

We now use the unbiasedness of our protocol to obtain a result for estimating the squared error between the estimated average vector and the true average vector.
When calculating the MSE, each coordinate location is used with expectation $n/d$.
Therefore, we define the \emph{normalized} MSE, or $\widehat{\textnormal{MSE}}$, as the normalization of the MSE by a factor of $(n/d)^2$.

\begin{thm} \label{thm:mse}
For any $d, n \in \mathbb{N}$, $\{ t \in \mathbb{N} \ | \ t \in [d] \}$, $\varepsilon < 6$ and $\delta \in \!\left( 0,1 \right]$, there exists a parameter $k$ such that $\mathcal{P}_{d, k, n, t}$ is $(\varepsilon, \delta)$-DP and
\[
\widehat{\textnormal{MSE}}(\mathcal{P}_{d, k, n, t}) =
\begin{cases}
\frac{2t d^{8/3} (14 \log(1/\delta) \log(2t/\delta))^{2/3}}{(1-\gamma)^{2} n^{5/3} \varepsilon^{4/3}}, \\ \quad \text{when}\ \varepsilon < 1 \\
\frac{8t d^{8/3} (63 \log(1/\delta) \log(2t/\delta))^{2/3}}{(1-\gamma)^{2} n^{5/3} \varepsilon^{4/3}}, \\ \quad \text{when}\ 1 \leq \varepsilon < 6, \\
\end{cases}
\]
where $\widehat{\textnormal{MSE}}$ denotes the squared error between the estimated average vector and the true average vector.
\end{thm}

\begin{proof}
We consider the $\sum_{l = 1}^{d} {\tt DeBias}(\hat{z}^{(l)})$ of $\mathcal{P}_{d, k, n, t}$ compared to the corresponding input $\sum_{j = 1}^{t} \sum_{i = 1}^{n} x_{i}^{(\alpha_{ij})}$ over the dataset $\vec{D}$.
We use the bounds on the variance of the randomized response mechanism from Theorem~\ref{thm:mse} to give us an upper bound for this comparison.

\begin{align*}
\operatorname{MSE}&(\mathcal{P}_{d, k, n, t}) \\
&= \sup_{\vec{D}} \E \!\left[ \!\left\| \sum_{l = 1}^{d} {\tt DeBias} ( \hat{z}^{(l)} ) \hspace{.5mm} e_{l} - \sum_{j = 1}^{t} \sum_{i = 1}^{n} x_{i}^{(\alpha_{ij})} \hspace{.5mm} e_{\alpha_{ij}} \right\|_{2}^{2} \ \right] \\
&\textnormal{(where} \ e_{l} \ \textnormal{is the} \ l\textsuperscript{th} \ \textnormal{basis vector)} \\
&= \sup_{\vec{D}} \E \!\left[ \!\left( \sum_{j = 1}^{t} \sum_{i = 1}^{n} \!\left( {\tt DeBias} ( y_{i}^{(\alpha_{ij})}/k ) - x_{i}^{(\alpha_{ij})} \right) \right)^{2} \ \right] \\
&= \sup_{\vec{D}} \sum_{j = 1}^{t} \sum_{i = 1}^{n} \E \!\left[ \!\left( {\tt DeBias} ( y_{i}^{(\alpha_{ij})}/k ) - x_{i}^{(\alpha_{ij})} \right)^{2} \ \right] \\
&\textnormal{(squared random variables are unbiased and independent)} \\
&= \sup_{\vec{D}} \sum_{j = 1}^{t} \sum_{i = 1}^{n} \operatorname{Var} \!\left[ {\tt DeBias} ( y_{i}^{(\alpha_{ij})}/k ) \right] \\
&= \frac{tn}{(1 - \gamma)^{2}} \sup_{x_{1}^{(\alpha_{i1})}} \text{Var} [ y_{1}^{(\alpha_{i1})}/k ]
\le \frac{tn}{(1 - \gamma)^{2}} \!\left( \frac{1 - \gamma}{4k^{2}} + \frac{\gamma}{2} \right) \\
&\le \frac{tn}{(1 - \gamma)^{2}} \!\left( \frac{1}{4k^{2}} + \frac{A_{\varepsilon} dk  \log(1/\delta)\log(2t/\delta)}{(n-1) \varepsilon^2} \right),
\end{align*}

\noindent where $A_{\varepsilon} = 28$ when $\varepsilon < 1$, and $A_{\varepsilon} = 1008$ when $1 \leq \varepsilon < 6$.
In other words, $A_{\varepsilon}$ is equal to half the constant term in the expression of $\gamma$ stated in Theorem~\ref{thm:main}.
The choice
$k = \frac{(n-1) \varepsilon^{2}}{4A_{\varepsilon} d \log(1/\delta) \log(2t/\delta)}$ minimizes the bracketed sum above and the bounds in the statement of the theorem follow.
\end{proof}

To obtain the error between the estimated average vector and the true average vector, we simply take the square root of the result obtained in Theorem~\ref{thm:mse}.

\begin{cor} \label{cor:query}
For every statistical query $q: \mathcal{X} \mapsto [0,1]^{d}$, $d, n \in \mathbb{N}$, $\{ t \in \mathbb{N} \ | \ t \in [d] \}$, $\varepsilon < 6$ and $\delta \in \!\left( 0,1 \right]$, there is an $(\varepsilon, \delta)$-DP $n$-party unbiased protocol for estimating $\frac{d}{n} \sum_{i} q(\vec{x}_{i})$ in the Single-Message Shuffle Model with standard deviation
\[
\hat{\sigma}(\mathcal{P}_{d, k, n, t}) =
\begin{cases}
\frac{(2t)^{1/2} d^{4/3} (14 \log(1/\delta) \log(2t/\delta))^{1/3}}{(1-\gamma) n^{5/6} \varepsilon^{2/3}}, \\ \quad \text{when}\ \varepsilon < 1 \\
\frac{(8t)^{1/2} d^{4/3} (63 \log(1/\delta) \log(2t/\delta))^{1/3}}{(1-\gamma) n^{5/6} \varepsilon^{2/3}}, \\ \quad \text{when}\ 1 \leq \varepsilon < 6, \\
\end{cases}
\]
where $\hat{\sigma}$ denotes the error between the estimated average vector and the true average vector.
\end{cor}

To summarize, we have produced a new unbiased protocol for the computation of the sum of $n$ real vectors in the Single-Message Shuffle Model with normalized MSE $O_{\varepsilon, \delta} (d^{8/3} t n^{-5/3})$, using advanced composition results from Dwork \emph{et al.}~\cite{dwork}.
Minimizing this bound as a function of $t$ leads us to choose $t=1$, but any choice of $t$ that is small and not dependent on $d$ produces a bound of the same order.
In our experimental study, we determine that the best choice of $t$ in practice is indeed $t=1$.

\subsection{Improved bounds for t=1} \label{sec:betterbounds}
We observe that in the optimal case in which $t=1$, we can tighten the bounds further, as we do not need to invoke the advanced composition results when each user samples only a single coordinate.
This changes the value of $\gamma$ by a factor of $O(\log(1/\delta))$, which propagates through to the expression for the MSE.
That is, we can more simply set $\varepsilon' = \varepsilon$ and $\delta' = \delta$ in the proof of Theorem~\ref{thm:main}.
When $\varepsilon < 1$, the computation is straightforward, with $c \geq \frac{14}{\varepsilon'^2} \log(2t/\delta)$ being chosen as before.
However, when $1 \leq \varepsilon < 6$, a tighter $c \geq \frac{80}{\varepsilon'^2} \log(2t/\delta)$ must be selected, as the condition $\varepsilon' < 1$ no longer holds.

Using $\varepsilon' < 6$, we have:
\[
(1 - \exp \hspace{.5mm} (-\varepsilon'/2)) \ge \left( 1 - \exp \left (-\frac{2}{3\sqrt{15}} \right) \right) \varepsilon' \ge \frac{\varepsilon'}{2\sqrt{10}}.
\]

\noindent Thus, we have:
\begin{align*}
\Pr \!\left[ \frac{\mathsf{N}_{\theta}}{\mathsf{N}_{\phi}}  \geq e^{\varepsilon'} \right]
&\le \exp \dftbig( -\frac{c}{3} (\varepsilon'/2)^2 \dftbig) +
\exp \dftbig( -\frac{c}{2} \dftbig( \frac{\varepsilon'}{2\sqrt{10}} \dftbig) ^2 \dftbig)\\
&\le 2\exp\left(-\frac{80}{2\varepsilon'^2}\frac{\varepsilon'^2}{40} \log(2t/\delta)\right) \leq \delta/t,
\end{align*}

\noindent which yields:
\[
\gamma =
\begin{cases}
\max \!\dftbig\{ \frac{14dk \log(2/\delta)}{(n - 1) \varepsilon^{2}}, \frac{27dk}{(n - 1) \varepsilon} \dftbig\}, & \text{when}\ \varepsilon < 1 \\
\max \!\dftbig\{ \frac{80dk \log(2/\delta)}{(n - 1) \varepsilon^{2}}, \frac{36dk}{11(n - 1) \varepsilon} \dftbig\}, & \text{when}\ 1 \leq \varepsilon < 6. \\
\end{cases}
\]

Note that the above expression for $\gamma$ in the case $\varepsilon < 1$ coincides with the result obtained by Balle \emph{et al.} in the scalar case~\cite{balleprivacyblanket}.
Putting this expression for $\gamma$ in the proof of Theorem~\ref{thm:mse}, with the choice
\[
k = 
\begin{cases}
\min \!\dftbig\{ \!\left( \frac{n \varepsilon^{2}}{28d \log(2/\delta)} \right)^{1/3}, \ \!\left( \frac{n \varepsilon}{54d} \right)^{1/3} \dftbig\}, \\ \qquad \text{when}\ \varepsilon < 1 \\
\min \!\dftbig\{ \!\left( \frac{n \varepsilon^{2}}{160d \log(2/\delta)} \right)^{1/3}, \ \!\left( \frac{11n \varepsilon}{72d} \right)^{1/3} \dftbig\}, \\ \qquad \text{when}\ 1 \leq \varepsilon < 6,
\\
\end{cases}
\]

\noindent causes the upper bound on the normalized MSE to reduce to:
\[
\widehat{\textnormal{MSE}} =
\begin{cases}
\max \!\dftbig\{ \frac{98^{1/3} d^{8/3} \log^{2/3}(2/\delta)}{(1-\gamma)^{2} n^{5/3} \varepsilon^{4/3}}, \frac{18 d^{8/3}}{(1-\gamma)^{2} n^{5/3} (4 \varepsilon)^{2/3}} \dftbig\}, \\ \qquad \text{when}\ \varepsilon < 1 \\
\max \!\dftbig\{ \frac{2 d^{8/3} (20\log(2/\delta))^{2/3}}{(1-\gamma)^{2} n^{5/3} \varepsilon^{4/3}}, \frac{2(9^{2/3}) d^{8/3}}{(1-\gamma)^{2} n^{5/3} (11 \varepsilon)^{2/3}} \dftbig\}, \\ \qquad \text{when}\ 1 \leq \varepsilon < 6. \\
\end{cases}
\]

By updating Corollary \ref{cor:query} in the same way, we can conclude that for the optimal choice $t = 1$, the normalized standard deviation of our unbiased protocol can be further tightened to:
\[
\hat{\sigma} =
\begin{cases}
\max \!\dftbig\{ \frac{98^{1/6} d^{4/3} \log^{1/3}(2/\delta)}{(1-\gamma) n^{5/6} \varepsilon^{2/3}}, \frac{18^{1/2} d^{4/3}}{(1-\gamma) n^{5/6} (4 \varepsilon)^{1/3}} \dftbig\}, \\ \qquad \text{when}\ \varepsilon < 1 \\
\max \!\dftbig\{ \frac{2^{1/2} d^{4/3} (20\log(2/\delta))^{1/3}}{(1-\gamma) n^{5/6} \varepsilon^{2/3}}, \frac{2^{1/2} 9^{1/3} d^{4/3}}{(1-\gamma) n^{5/6} (11 \varepsilon)^{1/3}} \dftbig\}, \\ \qquad \text{when}\ 1 \leq \varepsilon < 6. \\
\end{cases}
\]

\section{Transforming Summation in the Shuffle Model} \label{sec:transform}

In this section we further improve the bound we have obtained for private summation by using an orthonormal transformation.
We make use of the (Discrete) Fourier Transformation of the data, which concentrates information about signals with a particular property into a small number of coefficients.
We follow the outline of Rastogi and Nath~\cite{rastogi}, who follow a similar approach for time series data in the centralized DP model.
Our goal is to seek to improve the normalized MSE of our protocol, by concentrating on a smaller number of coefficients in the Fourier domain.

Recall that we are addressing the problem of computing the sum of $n$ real $d$-dimensional vectors, each of the form
\[
\vec{x}_{i} = (x_{i}^{(1)}, \dots, x_{i}^{(d)}) \in [0,1]^{d},
\] 
in the Single-Message Shuffle Model.
In Section~\ref{sec:protocol}, we formulated a new protocol $\mathcal{P}_{d, k, n, t}$, which adds random noise to each vector $\vec{x}_{i}$ in turn, ensuring that the computation of the (approximate) sum $\vec{z} = (z^{(1)}, \dots, z^{(d)}) \in [0,1]^{d}$ of these vectors is $(\varepsilon, \delta)$-DP.
In particular, a randomized response mechanism was applied to each of the $t$ uniformly sampled coordinates from the $d$ available choices.
In Section~\ref{sec:betterbounds}, we obtained our tight bound $O_{\varepsilon, \delta} (d^{8/3} n^{-5/3})$ for the normalized MSE of our protocol.

If we are able to compress each of the vectors $\vec{x}_{i}$ to a highly representative $m$-dimensional vector before applying $\mathcal{P}_{d, k, n, t}$, it will be possible to improve this bound to $O_{\varepsilon, \delta} (m^{8/3} n^{-5/3})$.
Our method involves applying the \emph{Discrete Fourier Transform} (DFT) to the $d$-dimensional vector $\vec{x}_{i}$ to obtain another $d$-dimensional vector.
The key to this approach is the assumption that the DFT captures the bulk of the information about the vector in a prefix of the coefficients.
While this is not true in general for arbitrary signals, such as ones where each component is chosen independently and uniformly at random, it has been observed to hold for many naturally occurring scenarios, such as time-series of human and natural activity, audio signals, and so on~\cite{selesnick}.
When this assumption holds, it is possible to eliminate most of the coefficients of the transformed vector whilst keeping the vast majority of the information about the data.
In particular, this holds true for the ECG Heartbeat Categorization Dataset that we use in our experimental study, as we see later.
Absent the above property, eliminating coefficients in this way would not necessarily result in most of the information being retained.

By keeping only the first $m$ Fourier coefficients of $\text{DFT}(\vec{x}_{i})$, where $m \ll d$, and then applying $\mathcal{P}_{d, k, n, t}$ to $m$ coefficients instead of $d$, we can ensure that the accuracy lost from the $d - m$ eliminated coordinates is much smaller than the improvement in the normalized MSE bound.
This close variant of $\mathcal{P}_{d, k, n, t}$ will be expressed as an algorithm $\mathcal{F}_{d, k, m, n, t}$ in Section~\ref{sec:thefsa}.
To motivate this, we first recall how to approximate a $d$-dimensional vector using the DFT and its inverse.

\subsection{Discrete Fourier Transform} \label{sec:dft}

The DFT of a $d$-dimensional vector $\vec{x}_{i} = (x_{i}^{(1)}, \dots, x_{i}^{(d)}) \\ \in [0,1]^{d}$ is defined to be the linear transform giving another $d$-dimensional vector $\text{DFT}(\vec{x}_{i}) = (\text{DFT} (x_{i}^{(1)}), \dots, \text{DFT} (x_{i}^{(d)})) \\ \in [0,1]^{d}$, where each $\text{DFT} (x_{i}^{(j)})$ coefficient is defined as:
\[
\text{DFT} (x_{i}^{(j)}) = \frac{1}{\sqrt{d}}\sum_{k = 1}^{d} x_{i}^{(k)} e^{\frac{2 \pi \sqrt{-1}}{d} jk}.
\]

The Inverse DFT of $\vec{x}_{i}$ is the corresponding inverse linear transform to the DFT.
It is represented as $\text{IDFT}(\vec{x}_{i}) = (\text{IDFT} (x_{i}^{(1)}), \dots, \text{IDFT} (x_{i}^{(d)})) \in [0,1]^{d}$, where each $\text{IDFT} (x_{i}^{(j)})$ is defined as:
\[
\text{IDFT} (x_{i}^{(j)}) = \frac{1}{\sqrt{d}} \sum_{k = 1}^{d} x_{i}^{(k)} e^{- \frac{2 \pi \sqrt{-1}}{d} jk}.
\]

Although the Fourier Transform gives complex results in general, the DFT can be represented by $d$ real numbers for real input data of dimension $d$.
Importantly, these real numbers can be bounded.
Given $\vec{y}$, we have $\|\vec{y}\|_2 = \|\text{DFT}(\vec{y})\|_2$ (Plancherel Theorem~\cite{herb}).
So if we ensure that our vectors are normalized so that $\|\vec{x}_i\|_1 = 1$, then $\|\text{DFT}(\vec{x}_i)\|_2 = \|\vec{x}_i\|_2 \leq \|\vec{x}_i\|_1$.
This in turn means that every $|\text{DFT}(\vec{x}_i^{(j)})|\leq 1$, i.e., the individual Fourier coefficient values are in the range $-1$ to $+1$.
An additional property is that the first Fourier coefficient gives the so-called `DC component', $\text{DFT}(\vec{x}_i^{(1)}) = \sum_{j=1}^{d}\vec{x}_i^{(j)}$, which, if $\vec{x}$ is a normalized non-negative vector, we can assume to be equal to 1.

We have established that in our case, each transformed vector $\text{DFT}(\vec{x}_{i})$ contains most of the information from the input.
So we can choose a small number $m \ll d$ such that only the first $m$ Fourier coefficients of the vector returned by $\text{DFT}(\vec{x}_{i})$ are kept.
This leaves an $m$-dimensional summary:
\[
\text{DFT}^{m} (\vec{x}_{i}) = (\text{DFT} (x_{i}^{(1)}), \dots, \text{DFT} (x_{i}^{(m)})) \in [-1,1]^{m}.
\]
We retrieve a version of the original data by `padding' the summary, by appending $d - m$ zeros to $\text{DFT}^{m} (\vec{x}_{i})$, denoted by $\text{PAD}^{d}$, then performing the inverse transform:
\[
\vec{x}'_{i} = (x_{i}^{(1)\prime}, \dots, x_{i}^{(d)\prime}) = \text{IDFT} (\text{PAD}^{d} (\text{DFT}^{m} (\vec{x}_{i}))).
\]
The accuracy of this approximation is calculated via the \emph{reconstruction error} of each coordinate:
\[
\text{RE}_{j}^{m} (\vec{x}_{i}) = \dftbig( x_{i}^{(j)\prime} - x_{i}^{(j)} \dftbig) ^{2} = \sum_{j=m+1}^d \text{DFT}(x_i^{(j)})^2.
\]

\subsection{Fourier Summation Algorithm} \label{sec:thefsa}

Algorithm~\ref{alg:thealgo} describes $\mathcal{F}_{d, k, m, n, t}$, an application of the approximation method from Section~\ref{sec:dft} to the private summation of vector-valued messages.
After the first $m$ Fourier coefficients in the DFT of each $\vec{x}_{i}$ are computed, we apply our protocol $\mathcal{P}_{d, k, n, t}$ from Section~\ref{sec:protocol} to each $m$-dimensional vector, where the analyzer returns a debiased $m$-dimensional vector representing the mean of the aggregated outputs from each user.
Note that in this algorithm each user randomizes $t$ uniformly sampled coordinates from their transformed vector, so their sample is likely to be much more representative of the original vector.
To complete the algorithm, the returned $m$-dimensional vector is `padded' with $d - m$ zeros and then transformed back to the original domain.
The output of $\mathcal{F}_{d, k, m, n, t}$ is a close approximation to the output of $\mathcal{P}_{d, k, n, t}$, differing only in the reconstruction errors of each returned coordinate.

There is one discrepancy to address: our basic vector summation protocol requires each coordinate to be in the range $[0,1]$, while the DFT values may be in the range $[-1,1]$.
There are two natural approaches.
We could extend the protocol to handle negative values, by expanding the the histogram to $2k$ buckets, $k$ for positive values and $k$ for the negative ones.
Or, we could remap the Fourier coefficients by a linear transformation (adding 1 and dividing the result by 2) before putting them into the protocol, then applying the inverse of this transform on the decoded result.
We apply the latter approach in our experiments.

The privacy of this procedure follows immediately from the discussion in Section~\ref{sec:privanalysis}.
The DFT of a vector of dimension $d$ produces a new vector of the same dimension, whose privacy is protected by the shuffle-based protocol.
The inversion of the DFT on the reconstructed vector can be considered as post-processing, and does not affect the privacy properties of the procedure.

\begin{algorithm}[t]
\DontPrintSemicolon
\SetArgSty{textnormal}
\smallskip
\KwPP{$k$, $m$, $t$, dimension $d$, and number of parties $n$}\\
\KwIn{$\vec{D} = (\vec{x}_{1}, \dots, \vec{x}_{n}) \in ([0,1]^{d})^{n}$}
\smallskip
Compute $\vec{D}^{*} = (\text{DFT}^{m} (\vec{x}_{1}), \dots, \text{DFT}^{m} (\vec{x}_{n})) \in ([0,1]^{m})^{n}$\\
\smallskip
Compute $\vec{z}^{*} = (\mathcal{P}_{d, k, n, t} (\text{DFT}^{m} (\vec{x}_{1})), \dots, \mathcal{P}_{d, k, n, t} (\text{DFT}^{m} (\vec{x}_{n}))) \in [0,1]^{m}$\\
\smallskip
Return $\vec{z}' = \text{IDFT} (\text{PAD}^{d} (\vec{z}^{*})) \in [0,1]^{d}$\\
\smallskip
\KwOut{$\vec{z}' = (z^{(1)\prime}, \dots, z^{(d)\prime}) \in [0,1]^{d}$}
\smallskip
\caption{Fourier Summation $\mathcal{F}_{d, k, m, n, t}$}
\label{alg:thealgo}
\end{algorithm}

\subsection{Analyzing Accuracy} \label{sec:accuracy}

In Section~\ref{sec:betterbounds} we refined the bound obtained from Theorem~\ref{thm:mse} to state that for any $d, n \in \mathbb{N}$, $t = 1$, $\varepsilon < 6$ and $\delta \in \!\left( 0, 1 \right]$, there exists a parameter $k$ such that $\mathcal{P}_{d, k, n, t}$ is $(\varepsilon, \delta)$-DP and:
\begin{align*}
\widehat{\textnormal{MSE}}&(\mathcal{P}_{d, k, n, t}) \\
&=
\begin{cases}
\max \!\dftbig\{ \frac{98^{1/3} d^{8/3} \log^{2/3}(2/\delta)}{(1-\gamma)^{2} n^{5/3} \varepsilon^{4/3}}, \frac{18 d^{8/3}}{(1-\gamma)^{2} n^{5/3} (4 \varepsilon)^{2/3}} \dftbig\}, \\ \qquad \text{when}\ \varepsilon < 1 \\
\max \!\dftbig\{ \frac{2 d^{8/3} (20\log(2/\delta))^{2/3}}{(1-\gamma)^{2} n^{5/3} \varepsilon^{4/3}}, \frac{2(9^{2/3}) d^{8/3}}{(1-\gamma)^{2} n^{5/3} (11 \varepsilon)^{2/3}} \dftbig\}, \\ \qquad \text{when}\ 1 \leq \varepsilon < 6. \\
\end{cases}
\end{align*}

As $\mathcal{F}_{d, k, m, n, t}$ applies $\mathcal{P}_{d, k, n, t}$ on $m$-dimensional vectors, we expect its normalized MSE to be a function of $m$ instead of $d$, plus the reconstruction error for using $m$ instead of $d$ Fourier coefficients.
Note that any $\gamma$ that guarantees $(\varepsilon, \delta)$-DP in $\mathcal{P}_{d, k, n, t}$ will also guarantee $(\varepsilon, \delta)$-DP in $\mathcal{F}_{d, k, m, n, t}$.
Using this information, we calculate the normalized MSE of $\mathcal{F}_{d, k, m, n, t}$ in the following theorem.

\begin{thm} \label{thm:dftmse}
Fix the value of $\gamma$ we found in Theorem~\ref{thm:main} so that $\mathcal{F}_{d, k, m, n, t}$ is $(\varepsilon, \delta)$-DP.
Then, for all $j \in [d]$:
\begin{align*}
\widehat{\textnormal{MSE}}&(\mathcal{F}_{d, k, m, n, t}) \\
&=
\begin{cases}
\max \!\dftbig\{ \frac{98^{1/3} m^{8/3} \log^{2/3}(2/\delta)}{(1-\gamma)^{2} n^{5/3} \varepsilon^{4/3}}, \frac{18 m^{8/3}}{(1-\gamma)^{2} n^{5/3} (4 \varepsilon)^{2/3}} \dftbig\}, \\ \qquad \text{when}\ \varepsilon < 1 \\
\max \!\dftbig\{ \frac{2 m^{8/3} (20\log(2/\delta))^{2/3}}{(1-\gamma)^{2} n^{5/3} \varepsilon^{4/3}}, \frac{2(9^{2/3}) m^{8/3}}{(1-\gamma)^{2} n^{5/3} (11 \varepsilon)^{2/3}} \dftbig\}, \\ \qquad \text{when}\ 1 \leq \varepsilon < 6 \\
\end{cases}
\\
&+ \sum_{j = 1}^{d} \textnormal{RE}_{j}^{m} (\vec{z}).
\end{align*}
\end{thm}

\begin{proof}
Let $\vec{z}' = (z^{(1)\prime}, \dots, z^{(d)\prime}) \in [0,1]^{d}$ be the $d$-dimensional vector returned by the $\mathcal{F}_{d, k, m, n, t}$ algorithm.
We can make use of the orthonormality of the Fourier Transform to express the error in reconstruction in terms of the error in the Fourier coefficients:
\begin{align*}
&\widehat{\text{MSE}}(\mathcal{F}_{d, k, m, n, t})
= (\vec{z}' - \vec{z})^2
= \sum_{j=1}^{d} \text{DFT}(z'^{(j)} - z^{(j)})^2\\
&= \sum_{j=1}^{m} \text{DFT}(z'^{(j)} - z^{(j)})^2
+ \sum_{j=m+1}^{d} \text{DFT}(z'^{(j)} - z^{(j)})^2\\
&= \widehat{\text{MSE}}(\mathcal{P}_{d=m, k, n, t}) + \sum_{j=m+1}^{d} \text{DFT}(z^{(j)})^2\\
&=
\begin{cases}
\max \!\dftbig\{ \frac{98^{1/3} m^{8/3} \log^{2/3}(2/\delta)}{(1-\gamma)^{2} n^{5/3} \varepsilon^{4/3}}, \frac{18 m^{8/3}}{(1-\gamma)^{2} n^{5/3} (4 \varepsilon)^{2/3}} \dftbig\}, \\ \qquad \text{when}\ \varepsilon < 1 \\
\max \!\dftbig\{ \frac{2 m^{8/3} (20\log(2/\delta))^{2/3}}{(1-\gamma)^{2} n^{5/3} \varepsilon^{4/3}}, \frac{2(9^{2/3}) m^{8/3}}{(1-\gamma)^{2} n^{5/3} (11 \varepsilon)^{2/3}} \dftbig\}, \\ \qquad \text{when}\ 1 \leq \varepsilon < 6 \\
\end{cases}
\\
&\qquad + \sum_{j = 1}^{d} \text{RE}_{j}^{m} (\vec{z}).
\end{align*}

\begin{align*}
&\widehat{\text{MSE}}(\mathcal{F}_{d, k, m, n, t})
= (\vec{z}' - \vec{z})^{2} \leq (\mu - \vec{z})^{2} + (\vec{z}' - \mu)^{2} \\
&\le \sum_{j = 1}^{d} (z^{(j)\prime} - z^{(j)})^{2} + \widehat{\text{MSE}}(\mathcal{P}_{d=m, k, n, t}) \\
&= \sum_{j = 1}^{d} \text{RE}_{j}^{m} (\vec{z}) \\
&\qquad + 
\begin{cases}
\max \!\dftbig\{ \frac{98^{1/3} m^{8/3} \log^{2/3}(2/\delta)}{(1-\gamma)^{2} n^{5/3} \varepsilon^{4/3}}, \frac{18 m^{8/3}}{(1-\gamma)^{2} n^{5/3} (4 \varepsilon)^{2/3}} \dftbig\}, \\ \qquad \text{when}\ \varepsilon < 1 \\
\max \!\dftbig\{ \frac{2 m^{8/3} (20\log(2/\delta))^{2/3}}{(1-\gamma)^{2} n^{5/3} \varepsilon^{4/3}}, \frac{2(9^{2/3}) m^{8/3}}{(1-\gamma)^{2} n^{5/3} (11 \varepsilon)^{2/3}} \dftbig\}, \\ \qquad \text{when}\ 1 \leq \varepsilon < 6. \\
\end{cases}
\end{align*}
\end{proof}

We also obtain a tighter bound for the analogous corollary to Theorem~\ref{thm:mse}.

\begin{cor} \label{cor:querymse}
For every statistical query $q: \mathcal{X} \mapsto [0,1]^{d}$, $d, n \in \mathbb{N}$, $t = 1$, $\varepsilon < 6$ and $\delta \in \!\left( 0,1 \right]$, there is an $(\varepsilon, \delta)$-DP $n$-party unbiased protocol for estimating $\frac{1}{n} \sum_{i} q(\vec{x}_{i})$ in the Single-Message Shuffle Model with standard deviation $\hat{\sigma}(\mathcal{F}_{d, k, m, n, t})$
\begin{align*}
&=
\begin{cases}
\max \!\dftbig\{ \frac{98^{1/6} m^{4/3} \log^{1/3}(2/\delta)}{(1-\gamma) n^{5/6} \varepsilon^{2/3}}, \frac{18^{1/2} m^{4/3}}{(1-\gamma) n^{5/6} (4 \varepsilon)^{1/3}} \dftbig\}, \\ \qquad \text{when}\ \varepsilon < 1 \\
\max \!\dftbig\{ \frac{2^{1/2} m^{4/3} (20\log(2/\delta))^{1/3}}{(1-\gamma) n^{5/6} \varepsilon^{2/3}}, \frac{2^{1/2} 9^{1/3} m^{4/3}}{(1-\gamma) n^{5/6} (11 \varepsilon)^{1/3}} \dftbig\}, \\ \qquad \text{when}\ 1 \leq \varepsilon < 6 \\
\end{cases}
\\
&\qquad + \sqrt{\sum_{j = 1}^{d} \textnormal{RE}_{j}^{m} (\vec{z})}.
\end{align*}
\end{cor}

\begin{figure*}[p!]
\centering
\subfloat[Experimental error by number of coordinates $t$ retained]{\includegraphics[width=0.4\linewidth]{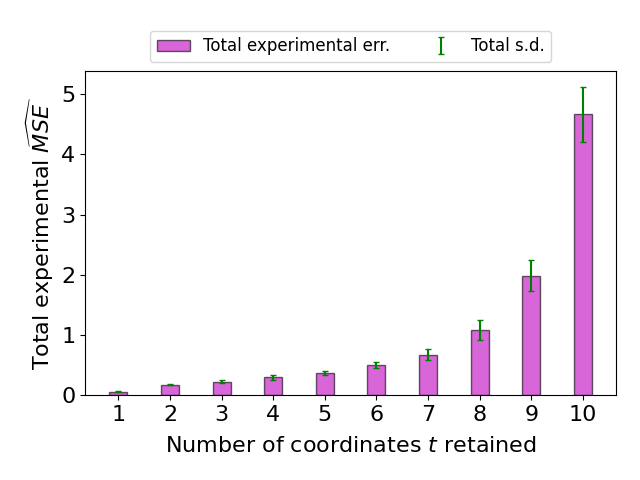}}%
\qquad\qquad
\subfloat[Experimental error by number of buckets $k$ used]{\includegraphics[width=0.4\linewidth]{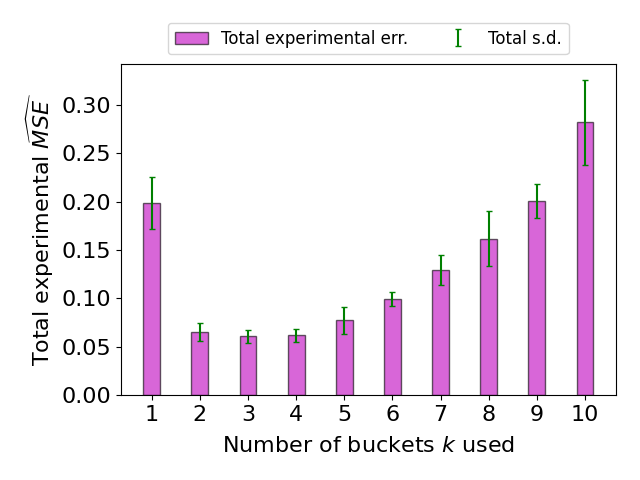}}%
\qquad\qquad
\subfloat[Experimental error by vector dimension $d$]{\includegraphics[width=0.4\linewidth]{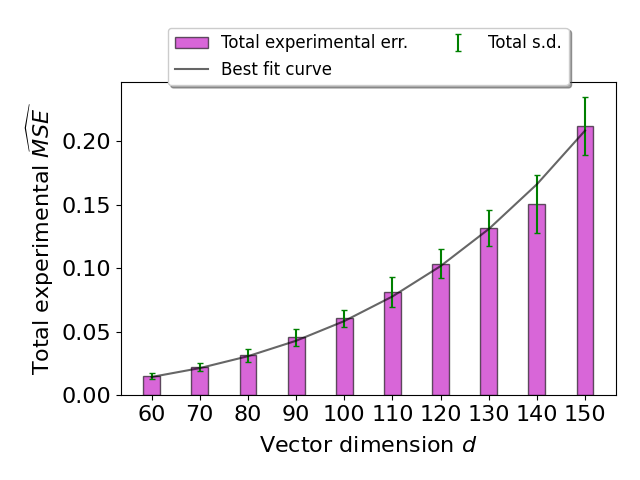}}%
\qquad\qquad
\subfloat[Experimental error by value of $\varepsilon$ where $0.5 \leq \varepsilon < 1$]{\includegraphics[width=0.4\linewidth]{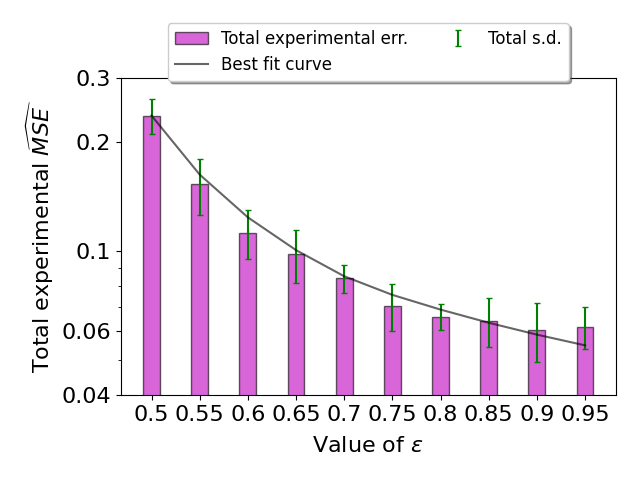}}%
\qquad\qquad
\subfloat[Experimental error by value of $\varepsilon$ where $1 \leq \varepsilon < 6$]{\includegraphics[width=0.4\linewidth]{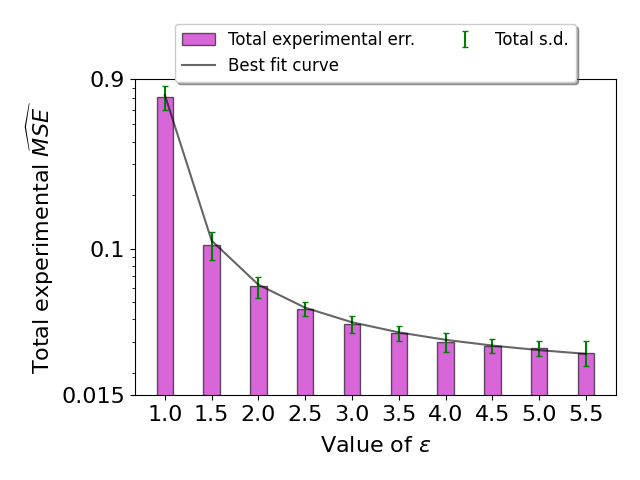}}%
\qquad\qquad
\subfloat[Experimental error by number of vectors $n$ used]{\includegraphics[width=0.4\linewidth]{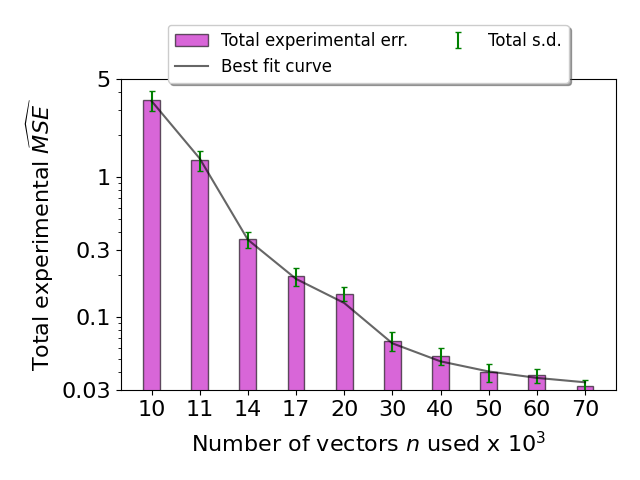}}\\
\caption{\small\label{fig:basic} Bar charts confirming that the choices $t = 1$ in (a) and $k = 3$ in (b) minimize the total experimental $\widehat{\textnormal{MSE}}$, and that best fit curves confirm the dependencies $d^{8/3}$ in (c), $\varepsilon^{-4/3}$ in (d) and (e), and $n^{-5/3}$ in (f) for the ECG Heartbeat Categorization Dataset in the non-Fourier case.}
\end{figure*}

\begin{figure*}[thb]
\centering
\subfloat[Experimental error by number of coordinates $t$ retained]{\includegraphics[width=0.4\linewidth]{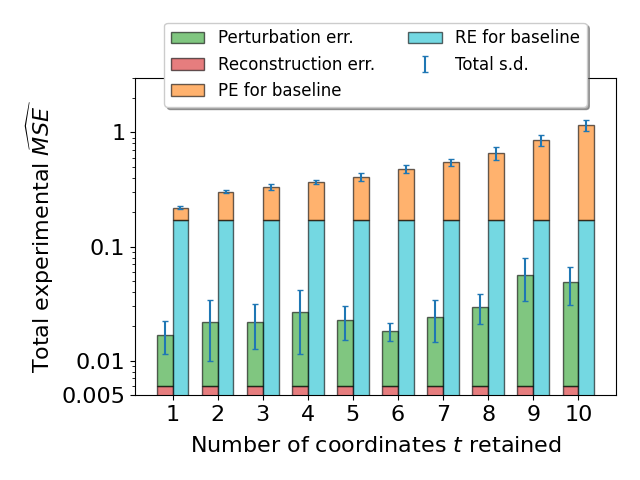}}%
\qquad\qquad
\subfloat[Experimental error by number of buckets $k$ used]{\includegraphics[width=0.4\linewidth]{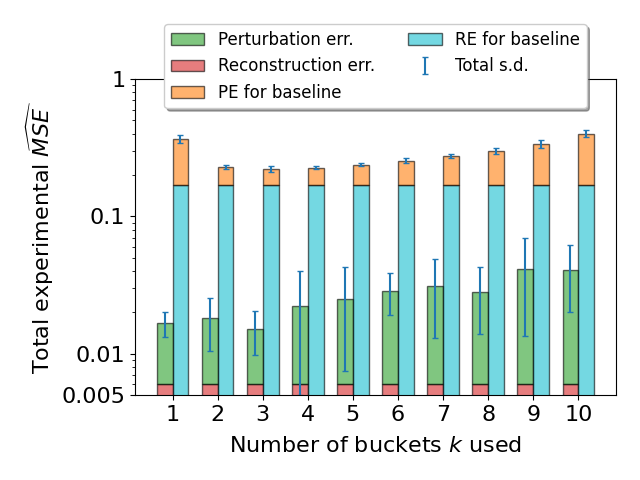}}\\
\caption{\small\label{fig:dfttk} Bar charts confirming that the choices $t = 1$ in (a) and $k = 3$ in (b) minimize the total experimental $\widehat{\textnormal{MSE}}$ for the ECG Heartbeat Categorization Dataset in the Fourier case.
The first bar originates from the authors' FSA, and the second an otherwise identical baseline case with the DFT removed.}
\end{figure*}

To summarize, we have improved the normalized MSE of our new unbiased protocol $\mathcal{P}_{d, k, n, t}$ for the computation of the sum of $n$ real vectors in the Single-Message Shuffle Model to $O_{\varepsilon, \delta} (m^{8/3} n^{-5/3})$, where $m$ can be much smaller than $d$, by using the DFT to compress each of the vectors to be $m$-dimensional, but retain most of their data.

To choose the right $m$, we need to find a good balance between the terms in Theorem~\ref{thm:dftmse}.
If $m$ is too big, the perturbation error $O (m^{8/3} n^{-5/3})$ gives the performance of $\mathcal{P}_{d, k, n, t}$, while if $m$ is too small the reconstruction error $\sum_{j = 1}^{d} \text{RE}_{j}^{m} (\vec{z})$ becomes too big.

If we compare the result of Theorem~\ref{thm:dftmse} with the refined version of Theorem~\ref{thm:mse}, we can see that the dependence on $\varepsilon$ and $n$ are the same.
However, there is a dependence on $m^{8/3}$ in the former, replaced by a dependence on $d^{8/3}$ in the latter, where $m$ is chosen to be smaller than $d$, and could be much smaller.
This vast improvement in the dependence of the dimension is counteracted by the reconstruction error in the Fourier approach, which will not be too large as long as $m$ is set appropriately.
To find the optimal value for $m$ for $\mathcal{F}_{d, k, m, n, t}$, we will compare these two theorems numerically, using a realistic dataset to calculate the dependencies and the reconstruction error.

\section{Experimental Evaluation} \label{sec:eeval}

In this section we present and compare the bounds generated by applying Algorithms~\ref{alg:fixedpoint}, \ref{alg:analyzer} and \ref{alg:thealgo} to an ECG Heartbeat Categorization Dataset in Python.
This publicly available dataset can be found at \url{https://www.kaggle.com/shayanfazeli/heartbeat}, and our Python code for all experiments is available at \url{https://github.com/mary-python/dft/blob/master/shuffle}.
Firstly, we analyse the effect of changing one key parameter at a time, whilst the others remain the same.
Our default settings are vector dimension $d = 100$, rounding parameter $k = 3$, number of users $n = 50000$, number of sampled coordinates $t = 1$, and differential privacy parameters $\varepsilon = 0.95$ and $\delta = 0.5$.
The ranges of all the above parameters have been adjusted to best display the dependencies, whilst simultaneously ensuring that the parameter $\gamma$ of the randomized response mechanism is always within its permitted range of $[0,1]$.

In the later experiments, where we explore the relationship between each of $\varepsilon$ and $n$ on the perturbation error in the Fourier case, it is useful to simultaneously explore a range of (Fourier) coefficients $m$ from $5$ to $95$ to see the effect of this change on the magnitude of the perturbation error.
To emphasize the benefit of using our new Fourier Summation Algorithm (FSA) on the experimental errors, we also implement an almost identical baseline alternative.
In our baseline case, we select our $m$ coefficients as in the FSA, but we do not apply the DFT, or indeed the Inverse DFT to generate the output vector from the padded vector.
All other steps, including the selection of $t$ coordinates from our $m$ coefficients, the linear transform in the original space between the ranges $[-1,1]$ and $[0,1]$, the rounding of the coordinates and the randomized response step, still take place.

\begin{figure*}[thb]
\centering
\subfloat[Perturbation error for $0.5 \leq \varepsilon < 1$ and $m = 5$]{\includegraphics[width=0.4\linewidth]{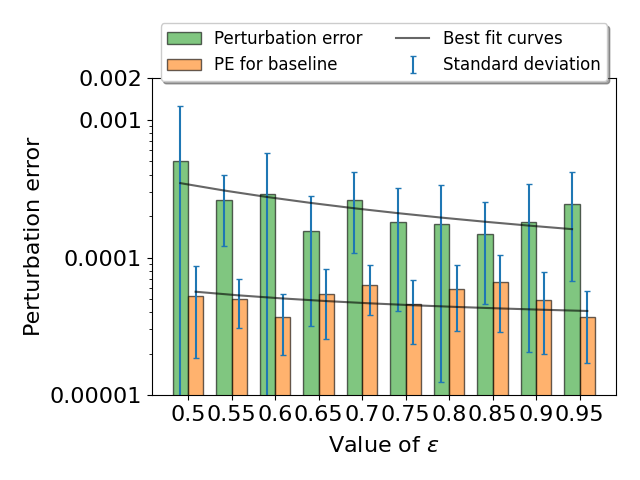}}%
\qquad\qquad
\subfloat[Perturbation error for $1 \leq \varepsilon < 6$ and $m = 5$]{\includegraphics[width=0.4\linewidth]{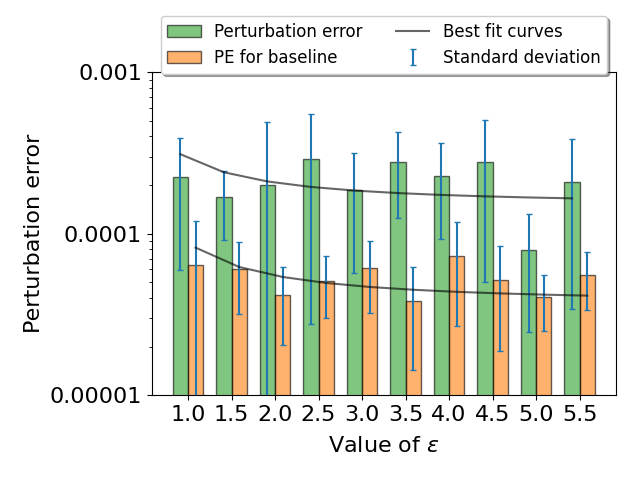}}%
\qquad\qquad
\subfloat[Perturbation error for $0.5 \leq \varepsilon < 1$ and $m = 95$]{\includegraphics[width=0.4\linewidth]{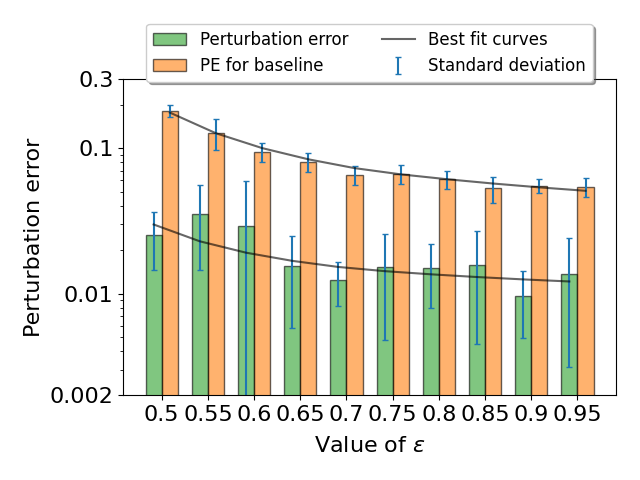}}%
\qquad\qquad
\subfloat[Perturbation error for $1 \leq \varepsilon < 6$ and $m = 95$]{\includegraphics[width=0.4\linewidth]{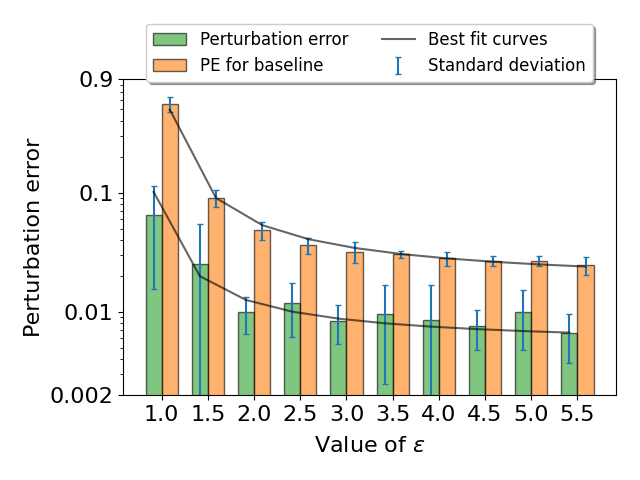}}\\
\caption{\small\label{fig:dfteps} Double bar charts displaying the effect of changing $\varepsilon$ for $m = 5$ and $m = 95$ on perturbation error.
Best fit curves confirm the dependency $\varepsilon^{-4/3}$ from Theorem~\ref{thm:dftmse}.
The first bar originates from the authors' FSA, and the second an otherwise identical baseline case with the DFT removed.}
\end{figure*}

\subsection{Results for Basic Protocol}

In the non-Fourier case (Algorithms~\ref{alg:fixedpoint} and \ref{alg:analyzer}), we first confirm that the choice of $t = 1$ is optimal, as predicted by the results of Section~\ref{sec:betterbounds}.
Indeed, Fig.~\ref{fig:basic} (a) shows that the total experimental $\widehat{\textnormal{MSE}}$ for the ECG Heartbeat Categorization Dataset is significantly smaller when $t = 1$, compared to any other small value of $t$, and so we adopt this setting in all further experiments.

Similarly, Fig.~\ref{fig:basic} (b) suggests that the total experimental $\widehat{\textnormal{MSE}}$ is lowest when $k = 3$, which is sufficiently close to the choice of $k$ selected in the proof of Theorem~\ref{thm:mse}, with all other default parameter values substituted in.
Observe that the absolute value of the observed MSE is below 0.3 in this case, meaning that the vector is reconstructed to a high degree of accuracy, sufficient for many applications.

\begin{figure*}[p!]
\centering
\subfloat[Perturbation error by value of $n$ when $m = 5$]{\includegraphics[width=0.4\linewidth]{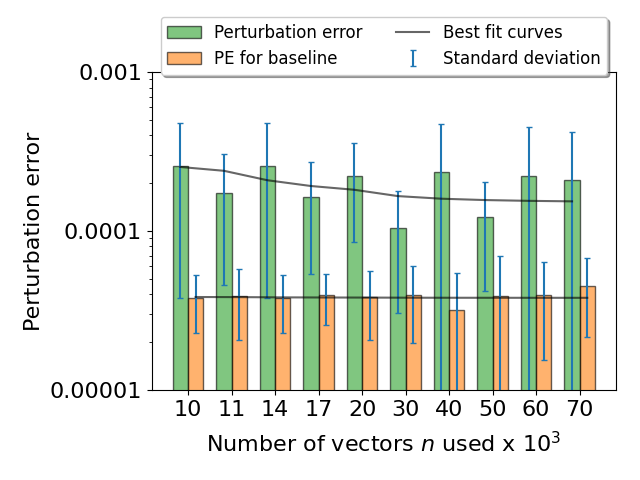}}%
\qquad\qquad
\subfloat[Perturbation error by value of $n$ when $m = 20$]{\includegraphics[width=0.4\linewidth]{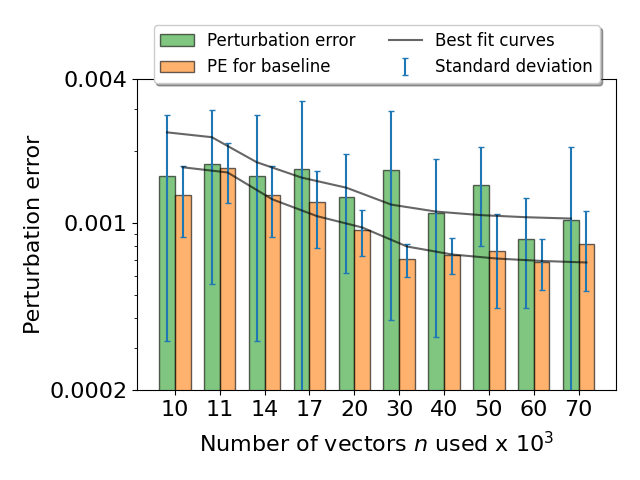}}%
\qquad\qquad
\subfloat[Perturbation error by value of $n$ when $m = 40$]{\includegraphics[width=0.4\linewidth]{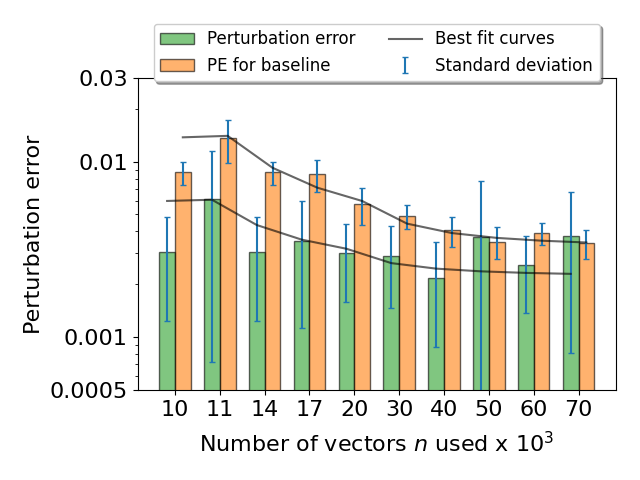}}%
\qquad\qquad
\subfloat[Perturbation error by value of $n$ when $m = 55$]{\includegraphics[width=0.4\linewidth]{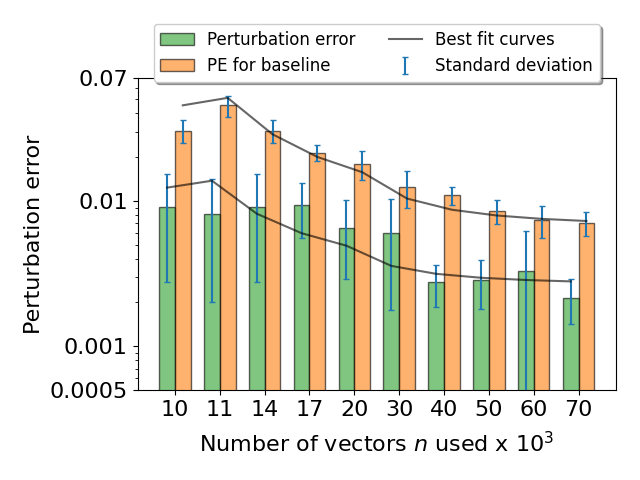}}%
\qquad\qquad
\subfloat[Perturbation error by value of $n$ when $m = 75$]{\includegraphics[width=0.4\linewidth]{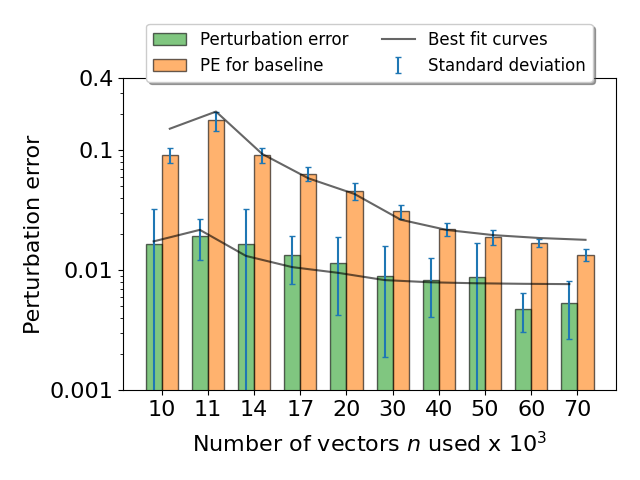}}%
\qquad\qquad
\subfloat[Perturbation error by value of $n$ when $m = 95$]{\includegraphics[width=0.4\linewidth]{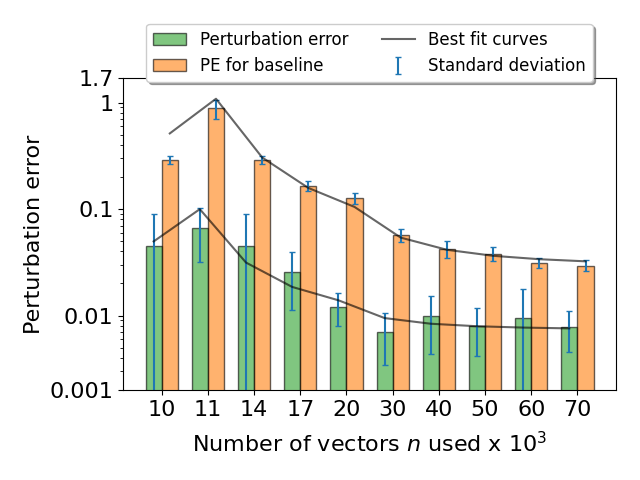}}\\
\caption{\small\label{fig:dftn} Double bar charts displaying the effect of changing $n$ for a range of values of $m$ on perturbation error.
Best fit curves confirm the dependency $n^{-5/3}$ from Theorem~\ref{thm:dftmse}.
The first bar originates from the authors' FSA, and the second an otherwise identical baseline case with the DFT removed.}
\end{figure*}

Next, we verify the bounds of $d^{8/3}$, $\varepsilon^{-4/3}$ and $n^{-5/3}$ from Theorem~\ref{thm:mse}.
Fig.~\ref{fig:basic} (c) is plotted with a best fit curve with equation a multiple of $d^{8/3}$, exactly as desired.
Unsurprisingly, the MSE increases as $d$ goes up according to this superlinear dependence.

Meanwhile, in Fig.~\ref{fig:basic} (d) and (e), we verify the dependency $\varepsilon^{-4/3}$ in the two ranges $\varepsilon < 1$ and $1 \leq \varepsilon < 6$.
The behavior for $\varepsilon < 1$ is quite smooth, but becomes more variable for larger $\varepsilon$ values.

A consequence of the way in which we ensure the privacy bounds are met for the range $1 \leq \varepsilon < 6$ is that 
the resulting experimental $\widehat{\textnormal{MSE}}$ in Fig.~\ref{fig:basic} (e) exceeds that for $\varepsilon=0.95$ in Fig.~\ref{fig:basic}~(d).
A tighter bound would be possible by separately considering these values of $\varepsilon$ when analyzing the term 
$1 - \exp(-\varepsilon'/2)$ (Section~\ref{sec:betterbounds}).
In the interests of brevity and not further overcomplicating the statement of the theoretical bounds, we omit this tightening.
A simpler fix is to replace $1 \leq \varepsilon \leq 1.5$ with $\varepsilon = 0.95$, to obtain both an improved accuracy and a stronger error guarantee.

We now look at Fig.~\ref{fig:basic} (f), which fits a curve dependent on $n^{-7/6}$, sufficiently close to the required result.
We see the benefit of increasing $n$: as $n$ increases by a factor of 10 across the plot, the error decreases by more than two orders of magnitude.

\begin{figure*}[thb]
\centering
\subfloat[Relationship between experimental errors for the ECG Heartbeat Categorization Dataset]{\includegraphics[width=0.4\linewidth]{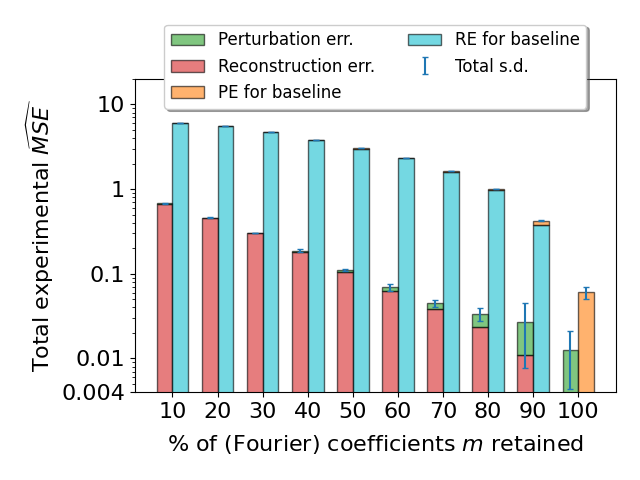}}%
\qquad\qquad
\subfloat[Relationship between experimental errors for a synthetic dataset created in Python]{\includegraphics[width=0.4\linewidth]{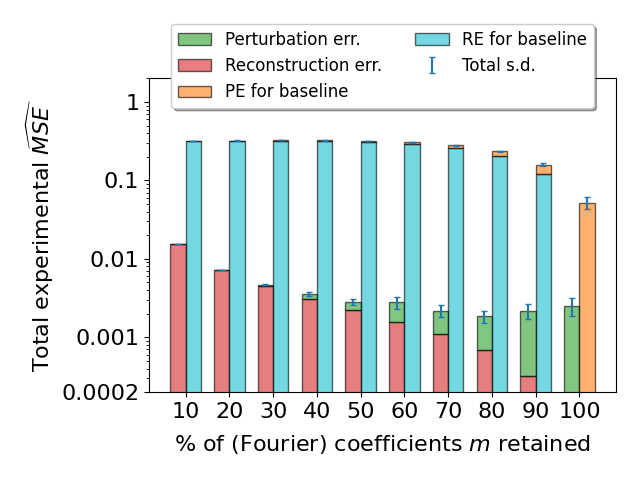}}%
\qquad\qquad
\subfloat[Perturbation error for the ECG Heartbeat Categorization Dataset, with best fit curve confirming $m^{8/3}$ dependency]{\includegraphics[width=0.4\linewidth]{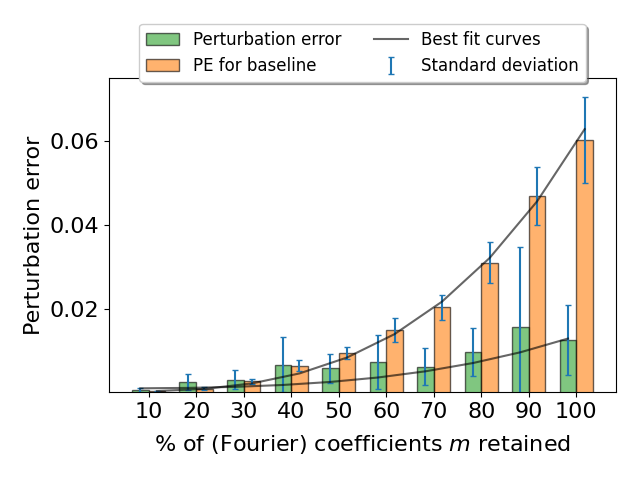}}%
\qquad\qquad
\subfloat[Perturbation error for a synthetic dataset created in Python, with best fit curve confirming $m^{8/3}$ dependency]{\includegraphics[width=0.4\linewidth]{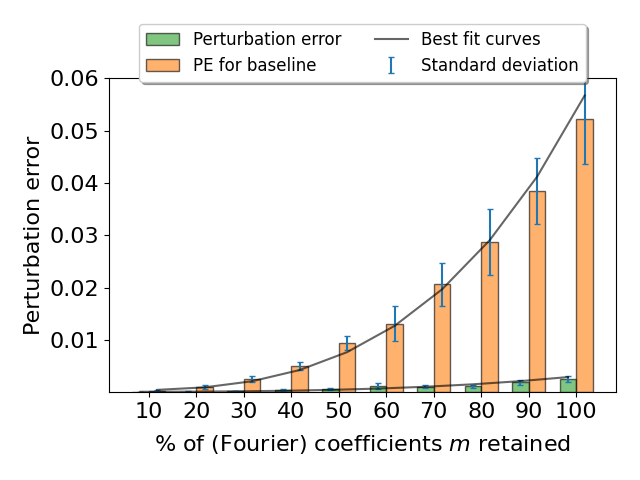}}\\
\caption{\small\label{fig:heartsynth} Double bar charts comparing the relationship between the perturbation and reconstruction errors in the Fourier case for (a) the ECG Heartbeat Categorization Dataset and (b) a synthetic dataset created in Python.
Perturbation error is isolated in (c) and (d) so that the $m^{8/3}$ dependency can be checked, and for ease of comparison.
The first bar originates from the authors' FSA, and the second an otherwise identical baseline case with the DFT removed.}
\end{figure*}

\subsection{Results for Fourier-based Protocol}

In the Fourier case (Algorithm \ref{alg:thealgo}), we used the packages \emph{`rfft'} and \emph{`irfft'} from SciPy's Fast Fourier Transform (FFT) module for the DFT and IDFT steps, which provided the most efficient computation with a real-valued output.
We compare the results of using the DFT to a baseline approach, in order to understand why the Fourier transform is well-suited to reducing the number of coefficients.
Our simple-minded baseline is to try to apply the same approach of dropping coordinates, but without the use of the Fourier transformation.
That is, we only consider the first $m$ coordinates of the input vector to apply the method of Section~\ref{sec:basic} to.
Our experiments demonstrate that this effort to reduce the dimensionality of the problem is clearly unsuccessful in comparison to the new Fourier Summation Algorithm (FSA).

In both the FSA and baseline cases, the preferred choices of $t = 1$ and $k = 3$ are confirmed in the same way as in the non-Fourier case, although the evidence is not quite as clear-cut.
The double bar charts in Fig.~\ref{fig:dfttk} (a) and (b) display the evidence for choosing $t = 1$ and $k = 3$ respectively.
We split the bars to show the reconstruction error due to using a fixed number of (Fourier) coefficients, and perturbation error, which comes from the randomness in the protocol.
It is clear to see that only the perturbation error is affected when $t$ or $k$ changes.

To check the dependencies $\varepsilon^{-4/3}$ and $n^{-5/3}$, the perturbation error must be separated, as Theorem~\ref{thm:dftmse} shows.
The perturbation error grows as we take more (Fourier) coefficients.
However, as we see in more detail below, this is outweighed by the reduction from reconstruction error, which pushes us towards picking a larger number of coefficients to minimize the total MSE.

In Fig.~\ref{fig:dfteps}, best fit curves proportional to $\varepsilon^{-4/3}$ have been plotted.
These curves fit the data quite well, as they pass through all but one of the error bars.
In a similar way, curves proportional to $n^{-5/3}$ confirm this remaining dependency in Fig.~\ref{fig:dftn}.
In this experiment, the reduction in error as $n$ increases is not as dramatic as the non-Fourier case.
However, increasing $n$ by a factor of 10 still reduces the error by more than an order of magnitude.

We now look more closely at the effect of changing the number of (Fourier) coefficients $m$ on the magnitude of the perturbation error, for our ranges of $\varepsilon$ and $n$ in Fig.~\ref{fig:dfteps}.
We first compare the $\varepsilon$ dependencies when $5\%$ of (Fourier) coefficients have been taken, with $95\%$ of (Fourier) coefficients.
It is clear from the (a), (b) and (c), (d) pairs in Fig.~\ref{fig:dfteps} that taking a very small number of (Fourier) coefficients results in a drastically smaller perturbation error, by at least two orders of magnitude.
We can see that the perturbation error for the FSA is consistently lower than for the baseline case when $m = 95$, however the opposite is true for $m = 5$.
As we will see later, the total experimental $\widehat{\textnormal{MSE}}$ for the FSA is always much smaller than the baseline case.
This is because, for small values of $m$, the huge reconstruction error in the baseline case outweighs any small changes in the already minuscule perturbation error.

A similar story can be seen in Fig.~\ref{fig:dftn}, where we explore four additional intermediate choices of (Fourier) coefficients, ranging from $20\%$ to $75\%$.
Increasing the (Fourier) coefficients fourfold from $5\%$ to $20\%$ increases the perturbation error by at least an order of magnitude, but the same is true for the lesser increases from $20\%$ to $55\%$, and from $55\%$ to $95\%$.
This shows that as the number of (Fourier) coefficients increases, the sensitivity of the perturbation error increases.
Therefore, it is important to choose a low number of Fourier coefficients to reduce perturbation error, but it does not have to be lower than $m = 20$, for example, as there is also a trade-off with reconstruction error.

We now include the reconstruction error once again to investigate the effect of changing the number of (Fourier) coefficients $m$ on the ratio between the perturbation and reconstruction errors.
To illustrate this pattern more clearly, we plot a graph using a randomly generated synthetic dataset with a sinusoidal dependence on each coordinate, as well as the ECG Heartbeat Categorization Dataset used in all the other experiments.
We also isolate the perturbation error in a separate graph for each dataset, for ease of comparison between the FSA and baseline cases.
All of these graphs are displayed together in Fig.~\ref{fig:heartsynth}.

As mentioned earlier in this section, we can see that for the ECG Heartbeat Categorization Dataset, the reconstruction error outweighs the perturbation error, preventing the pattern for the perturbation error to be seen clearly.
However, in the case of the synthetic dataset, the reconstruction error is much smaller, allowing the exponential increase of the perturbation error to have an impact on the total experimental $\widehat{\textnormal{MSE}}$.
We can see that when using the synthetic dataset, retaining approximately $80\%$ of the (Fourier) coefficients optimizes the total experimental $\widehat{\textnormal{MSE}}$, and this occurs soon after the perturbation error outweighs the reconstruction error.

Note that in all of the graphs in Fig.~\ref{fig:heartsynth}, the perturbation and reconstruction errors when the FSA is implemented are at least an order of magnitude smaller than the same errors in the baseline case.
The only exception is the perturbation error when the number of (Fourier) coefficients is low, but in that case the difference is not significant, especially compared to the magnitude of the corresponding reconstruction error.

In conclusion, these experiments confirm that picking $t = 1$ and $k = 3$ serves to minimize the error.
The lines of best fit confirm the dependencies on the other parameters from Sections~\ref{sec:vectorsum} and \ref{sec:transform} for $m$, $d$, $\varepsilon$ and $n$, by implementing and applying Algorithms~\ref{alg:fixedpoint}, \ref{alg:analyzer} and \ref{alg:thealgo} to an ECG Heartbeat Categorization Dataset in Python.
The experiments demonstrate that the MSE observed in practice is sufficiently small to allow effective reconstruction of average vectors for a suitably large cohort of users.

By comparing the implementation of our new Fourier Summation Algorithm (FSA) with a suitable baseline, we have demonstrated that our usage of the Discrete Fourier Transform (DFT) reduces all experimental errors significantly, regardless of the settings of all other parameters.

\section{Conclusion} \label{sec:conc}

Our results extend a result from Balle \emph{et al.}~\cite{balleprivacyblanket} for scalar sums to provide a new protocol $\mathcal{P}_{d, k, n, t}$ in the Single-Message Shuffle Model for the private summation of vector-valued messages $(\vec{x}_{1}, \dots, \vec{x}_{n}) \in ([0,1]^{d})^{n}$.
It is not surprising that the normalized MSE of the resulting estimator has a dependence on $n^{-5/3}$, as this was the case for scalars, but the addition of a new dimension $d$ introduces a new dependency for the bound, as well as the possibility of sampling $t$ coordinates from each $d$-dimensional vector.
For this extension, we formally defined the \emph{vector view} as the knowledge of the analyzer upon receiving the randomized vectors, and expressed it as a union of overlapping scalar views.
Through the use of advanced composition results from Dwork \emph{et al.}~\cite{dwork}, we showed that the estimator now has normalized MSE $O_{\varepsilon, \delta} (d^{8/3} t n^{-5/3})$ which can be further improved to $O_{\varepsilon, \delta} (d^{8/3} n^{-5/3})$ by setting $t = 1$.

To further improve this bound, we adapted the method of Rastogi \emph{et al.}~\cite{rastogi} to implement a \emph{Discrete Fourier Transform} (DFT).
The purpose of this method was to compress each of the $d$-dimensional vectors $\vec{x}_{i}$ to a highly representative $m$-dimensional vector, where $m \ll d$, and then apply $\mathcal{P}_{d, k, n, t}$ to $m$ coefficients instead of $d$.
Although some accuracy is lost by transforming the vectors between the original and Fourier domains, this is counteracted by the improvement in the normalized MSE from a dependence on $d^{8/3}$ to $m^{8/3}$.

Our contributions have provided a stepping stone between the summation of the scalar case discussed by Balle \emph{et al.}~\cite{balleprivacyblanket} and the linearization of more sophisticated structures such as matrices and higher-dimensional tensors, both of which are reliant on the functionality of the vector case.
We have seen via both theory (Section~\ref{sec:transform}) and experiments (Section~\ref{sec:eeval}) that combining our new private summation protocol with a DFT reduces the MSE significantly.

The work we have presented here may be elaborated in further work.
For example, a useful property of the Fourier space is that a convolution in normal space is equivalent to simple multiplication in Fourier space.
Although this property is typically used to improve speed, it could be explored as to whether this might be leveraged to gain additional privacy.
Further, as mentioned in Section~\ref{sec:litreview}, there is potential for further exploration in the Multi-Message Shuffle Model to gain additional privacy, by utilizing methods presented by Balle \emph{et al.}~\cite{ballemulti}.

\appendix
\section*{Proof of Lemma~\ref{lem:scalar}} \label{app:proof}
\primelemma*

\begin{proof}
The way in which we split the vector view (i.e., to consider a single uniformly sampled coordinate of each vector-valued message in turn), means that we can apply a proof that is analogous to the scalar-valued case~\cite{balleprivacyblanket}.
We work through the key steps needed.

Recall from Section~\ref{sec:basic} that the case where the $n^{\text{th}}$ user submits a uniformly random message independent of their input satisfies DP trivially.
Otherwise, the $n^{\text{th}}$ user submits their true message, and we assume that analyzer removes from $\vec{Y}^{(\alpha_{ij})}$ any truthful messages associated with the first $n - 1$ users.
Denote $n_{l}^{(\alpha_{ij})}$ to be the count of $j^{\text{th}}$ coordinates remaining with a particular value $l \in [k]$.
If $\vec{x}_{n}^{(\alpha_{ij})} = \theta$ and $\vec{x}_{n}^{\prime (\alpha_{ij})} = \phi$, we obtain the relationship
\[
\frac{\Pr [ \text{View}_{\mathcal{M}}^{(\alpha_{ij})}(\vec{D}) = V_{\alpha_{ij}} ] }{\Pr [ \text{View}_{\mathcal{M}}^{(\alpha_{ij})}(\vec{D}') 
= V_{\alpha_{ij}} ] } = \frac{n_{\theta}^{(\alpha_{ij})}}{n_{\phi}^{(\alpha_{ij})}}.
\]

\noindent We observe that the counts $n_{\theta}^{(\alpha_{ij})}$ and $n_{\phi}^{(\alpha_{ij})}$ follow the binomial distributions $\mathsf{N}_\theta \sim {\tt Bin} \dftbig( s, \frac{\gamma}{k} \dftbig) + 1$ and $\mathsf{N}_\phi \sim {\tt Bin} \dftbig( s, \frac{\gamma}{k} \dftbig) $ respectively, where $s$ denotes the number of times that the coordinate $j$ is sampled.
In expectation, $s = (n-1)t/d$, and below we will show that it is close to its expectation:

\begin{align*}
&\Pr_{\mathsf{V}_{\alpha_{ij}} \sim \text{View}_{\mathcal{M}}^{(\alpha_{ij})}(\vec{D})} \!\left[ \frac{ \Pr [ \text{View}_{\mathcal{M}}^{(\alpha_{ij})}(\vec{D}) = \mathsf{V}_{\alpha_{ij}} ] }
{ \Pr [ \text{View}_{\mathcal{M}}^{(\alpha_{ij})}(\vec{D}') = \mathsf{V}_{\alpha_{ij}} ] } \geq e^{\varepsilon'} \right]\\
&= \Pr \!\left[ \frac{\mathsf{N}_{\theta}}{\mathsf{N}_{\phi}} \geq e^{\varepsilon'} \right].
\end{align*}

\noindent We define $c:= \E[\mathsf{N}_{\phi}] = \frac{\gamma}{k} \cdot s$ and split this into the union of two events, $N_\theta \geq c e^{\varepsilon'/2}$ and $N_\phi \leq c e^{-\varepsilon'/2}$.
Applying a Chernoff bound gives:
\begin{align*}
\Pr \!\left[ \frac{\mathsf{N}_{\theta}}{\mathsf{N}_{\phi}} \geq e^{\varepsilon'} \right] &\le \exp \!\left( - \frac{c}{3} \!\left( e^{\varepsilon'/2} - 1 - \frac{1}{c} \right)^{2} \right)\\
&+ \exp \!\left( - \frac{c}{2} \!\left( 1 - e^{- \varepsilon'/2} \right)^{2} \right).
\end{align*}

\noindent We will choose $c \geq \frac{14}{\varepsilon'^2} \log(2t/\delta)$ so that we have:
\[
\exp \hspace{.5mm} (\varepsilon'/2) - 1 - \frac{1}{c}
\ge \frac{\varepsilon'}{2} + \frac{\varepsilon'^2}{8} - \frac{\varepsilon'^2}{14\log(2t/\delta)}
\ge \frac{\varepsilon'}{2}.
\]

\noindent Using $\varepsilon' < 1$, we have:
\[
(1 - \exp \hspace{.5mm} (-\varepsilon'/2)) \ge (1 - \exp \hspace{.5mm} (-1/2))\varepsilon' \ge \frac{\varepsilon'}{\sqrt{7}}.
\]

\noindent Thus we have:
\begin{align*}
\Pr \!\left[ \frac{\mathsf{N}_{\theta}}{\mathsf{N}_{\phi}}  \geq e^{\varepsilon'} \right]
&\le \exp \dftbig( -\frac{c}{3} (\varepsilon'/2)^2 \dftbig) +
\exp \dftbig( -\frac{c}{2}(\varepsilon'/\sqrt{7})^2 \dftbig)\\
&\le 2\exp\left(-\frac{14}{2\varepsilon'^2}\frac{\varepsilon'^2}{7} \log(2t/\delta)\right) \leq \delta/t.
\end{align*}

We now apply another Chernoff bound to show that $s \leq 2\E[s]$, which can be used to give a bound on $\gamma$.
The following calculation proves that $\Pr[s \geq 2\E(s)] \leq \exp(-\E(s)/3)$,
using $\E(s) = (n-1)t/d$:
\[
\Pr[s \geq 2\E(s)]
\leq  \exp \dftbig( -\frac{n-1}{3}t/d \dftbig) \leq  \exp \dftbig( -\frac{n}{3} \dftbig) < \delta/3t,
\]
for all reasonable values of $\delta$.

Substituting these bounds on $s$ and $c$ into $\gamma s/k = c$ along with
$\varepsilon'=\frac{\varepsilon}{2\sqrt{2t \log (1/\delta)}}$
gives:
\begin{align*}
\gamma \geq & \frac{112kt \log(1/\delta) \log(2t/\delta)}{s \varepsilon^{2}} 
\geq \frac{56dk \log(1/\delta)\log(2t/\delta)}{(n-1) \varepsilon^{2}}.
\end{align*}
\end{proof}

\end{document}